\documentclass[11pt]{article}
\usepackage[letterpaper,hmargin=1in,vmargin=1.25in]{geometry}

\usepackage{url}

\usepackage{amsmath}
\usepackage{amssymb}

\usepackage{amsthm}

\theoremstyle{plain}
\newtheorem{theorem}{Theorem}[section]
\newtheorem{lemma}[theorem]{Lemma}
\newtheorem{corollary}[theorem]{Corollary}
\newtheorem{proposition}[theorem]{Proposition}
\newtheorem{definition}[theorem]{Definition}

\newtheorem{conjecture}{Conjecture}

\renewcommand{\labelenumi}{(\roman{enumi})}

\newcommand{\abs}[1]{\left\lvert#1\right\rvert}

\newcommand{\rest}[2]{#1\!\!\restriction_{#2}}
\newcommand{\reste}[2]{#1\restriction_{#2}}
\newcommand{\osg}[1]{\left[#1\right]^{\prec}}

\newcommand{\N}{\mathbb{N}}%
\newcommand{\Z}{\mathbb{Z}}%
\newcommand{\Q}{\mathbb{Q}}%
\newcommand{\R}{\mathbb{R}}%
\newcommand{\X}{\{0,1\}^*}%
\newcommand{\XI}{\{0,1\}^\infty}%

\newcommand{\Lm}[1]{\mathcal{L}\left(#1\right)}

\DeclareMathOperator{\Prob}{Pr}

\newcommand{\noi}{\noindent}

\title{\textbf{Cryptography and Algorithmic Randomness%
\thanks{A part of this work was presented at the Seventh International Conference on Computability, Complexity and Randomness (CCR 2012),
July 2-6, 2012, Cambridge, Great Britain.}}}

\author{
Kohtaro Tadaki\hspace*{15mm}
Norihisa Doi\\
\\
Research and Development Initiative, Chuo University\\
1--13--27 Kasuga, Bunkyo-ku, Tokyo 112-8551, Japan\\
E-mail: \textsf{tadaki@kc.chuo-u.ac.jp},\quad \textsf{doi@doi.ics.keio.ac.jp}\\
\url{http://www2.odn.ne.jp/tadaki/}}

\date{}

\begin{document}

\maketitle

\begin{quotation}
\noi\textbf{Abstract.}
The secure instantiation of the random oracle is one of the major open problems in modern cryptography.
We investigate this problem using concepts and methods of algorithmic randomness.

In modern cryptography, the random oracle model is widely used as an imaginary framework in which the security of a cryptographic scheme is discussed.
In the random oracle model, the cryptographic hash function used in a cryptographic scheme is formulated
as a random variable uniformly distributed over all possibility of the function, called the random oracle.
The main result of this paper is to show that, for any secure signature scheme in the random oracle model,
there exists a specific computable function which can instantiate the random oracle while keeping the security originally proved in the random oracle model.
In modern cryptography the generic group model is used also for a similar purpose to the random oracle model.
We show that the same results hold for the generic group model.

In the process of proving the results, we introduce the notion of effective security, demonstrating the importance of this notion in modern cryptography.
\end{quotation}

\begin{quotation}
\noi\textit{Key words\/}:
cryptography,
random oracle model,
generic group model,
provable security,
algorithmic randomness,
computable analysis
\end{quotation}

\vspace*{0mm}

\section{Introduction}

In modern cryptography, the \emph{random oracle model} \cite{BR93}
is widely used as an \emph{imaginary} framework in which the security of a cryptographic scheme is discussed.
In the random oracle model, the cryptographic hash function used in a cryptographic scheme is formulated
as a random variable uniformly distributed over all possibility of the function,
called the \emph{random oracle}, and the legitimate users and the adversary against the scheme are modeled so as to get the values of the hash function
not by evaluating it in their own but by querying
the random oracle.
Since the random oracle is an imaginary object, even if the security of a cryptographic scheme is proved in the random oracle model,
the random oracle has to be instantiated using a concrete cryptographic hash function such as the SHA hash functions if we want to use the scheme in the real world.
In fact, the instantiations of the random oracle by concrete cryptographic hash functions are widely
used
in modern cryptography to produce efficient cryptographic schemes.
Once the random oracle is instantiated, however, the original security proof in the random oracle model is spoiled and
goes back to square one.
Actually, it is not clear how much the instantiation can maintain the security originally proved in the random oracle model,
nor is it clear whether the random oracle can be instantiated somehow while keeping the original security.

The question of securely instantiating the random oracle within cryptographic schemes proven secure in the random oracle model is
one of the most intriguing problems in modern cryptography.
Actually, many researches on the secure instantiation of the random oracle have been
done
so far, which include
Canetti, et al.~\cite{CGH04}, Bellare, et al.~\cite{BBP04}, Leurent and Nguyen~\cite{LN09}, Fischlin, et al.~\cite{FLRSST10}.
These mainly give negative results.

In this paper
we investigate
the problem of secure instantiation of the random oracle,
using concepts and methods of \emph{algorithmic randomness}.
Algorithmic randomness, also known as \emph{algorithmic information theory}, enables us to consider
the randomness of an \emph{individual} object.
It originated in the groundbreaking works of Solomonoff~\cite{Solom64}, Kolmogorov~\cite{Kol65}, and Chaitin~\cite{C66} in the mid-1960s.
They independently introduced the notion of \emph{program-size complexity}, also known as \emph{Kolmogorov complexity},
in order to quantify the randomness of an individual object.
Around the same time, Martin-L\"of \cite{M66} introduced a measure theoretic approach to characterize the randomness of an individual infinite binary sequence.
This approach, called \emph{Martin-L\"of randomness} nowadays, is one of the major notions in algorithmic randomness as well as program-size complexity.
Later on, in the 1970s
Schnorr~\cite{Sch73} and Chaitin~\cite{C75} showed that Martin-L\"of randomness is equivalent to the randomness defined by program-size complexity
in characterizing
random infinite binary sequences.
In the 21st century, algorithmic randomness makes remarkable progress through close interaction with recursion theory \cite{N09,DH10}.

In cryptography, the randomness is just a probability distribution or its
sequence.
Namely, the \emph{true randomness} in cryptography is a uniform probability distribution such as the random oracle,
while the \emph{pseudorandomness} is
a sequence of probability distributions which has a certain asymptotic property
defined based on computational complexity theory.
Thus, cryptology seems to have had no concern with the randomness of an individual object so far.
In algorithmic randomness, on the other hand, the notion of a \emph{random real} plays a central role.
It is an individual infinite binary sequence which is classified as ``random'', and not a random variable, unlike in cryptography.
Algorithmic randomness enables us to classify an individual infinite binary sequence into random or not.

To summarize our contributions, we first review the security proof in the random oracle model
(see e.g.~Katz and Lindell \cite[Chapter 13]{KL07} for the detail).
In the random oracle model, a cryptographic scheme $\Pi$ relies on an oracle $h$ which is a certain type of function mapping finite strings to finite strings,
depending on a security parameter $n$.
Let $\mathsf{Hash}_n$ denote the set of all such functions $h$ on
a
security parameter $n$.
Then the random oracle is the sequence $\{H_n\}$ of random variables such that
each $H_n$ is uniformly distributed over functions in $\mathsf{Hash}_n$.
Now, in order to introduce a security notion, such as CCA-security for encryption schemes and EUF-ACMA security for signature schemes, into the scheme $\Pi$,
we first consider an appropriately designed experiment $\mathsf{Expt}_{\mathcal{A}^{H_n},\mathrm{\Pi}^{H_n}}$ defined for
the scheme $\Pi$ and any adversary $\mathcal{A}$,
where $\Pi$ and $\mathcal{A}$ are both allowed to have an oracle access to
$H_n$.
Then a definition of security for $\Pi$ in the random oracle model takes the following general form:
the scheme $\Pi$ is \emph{secure in the random oracle model}
if, for all probabilistic polynomial-time adversaries $\mathcal{A}$ and all $d\in\N^+$ there exists $N\in\N^+$ such that, for all $n\ge N$,
\begin{equation}\label{gfsrom}
  \Prob\left[\mathsf{Expt}_{\mathcal{A}^{H_n},\mathrm{\Pi}^{H_n}}(n)=1\right]\le \gamma+\frac{1}{n^d},
\end{equation}
where the probability is taken over the random variable $H_n$, i.e., the random choice of a function in $\mathsf{Hash}_n$,
as well as the random choices of the parties running $\Pi$ and those of the adversary $\mathcal{A}$.
The value $\gamma$ indicates the maximum desired probability of some ``bad'' event (e.g., for encryption schemes $\gamma=1/2$ and for signature schemes $\gamma=0$).
Since the random variable $H_n$ is uniformly distributed over $\mathsf{Hash}_n$ for every $n$,
the definition \eqref{gfsrom} of security in the random oracle model is equivalently rewritten into the following form:
for all probabilistic polynomial-time adversaries $\mathcal{A}$ and all $d\in\N^+$ there exists $N\in\N^+$ such that, for all $n\ge N$,
\begin{equation}\label{rw-gfsrom}
  \frac{1}{\#\mathsf{Hash}_n}\sum_{h\in\mathsf{Hash}_n}
  \Prob\left[\mathsf{Expt}_{\mathcal{A}^{H_n},\mathrm{\Pi}^{H_n}}(n)=1\bigm|H_n=h\right]\le \gamma+\frac{1}{n^d},
\end{equation}
where
$\#\mathsf{Hash}_n$ denotes the number of functions in $\mathsf{Hash}_n$,
and the probability is now conditioned on that the random variable $H_n$ takes a specific function $h\in\mathsf{Hash}_n$ as its value.

Let $\{h_n\}$ be an arbitrary sequence of functions such that $h_n\in\mathsf{Hash}_n$ for all $n$.
In this paper, we introduce the notion of \emph{security of $\Pi$ relative to a specific oracle $\{h_n\}$}, which can be formulated as follows:
the scheme $\Pi$ is secure relative to $\{h_n\}$
if, for all probabilistic polynomial-time adversaries $\mathcal{A}$ and all $d\in\N^+$ there exists $N\in\N^+$ such that, for all $n\ge N$,
\begin{equation}\label{gfsrom-rel}
  \Prob\left[\mathsf{Expt}_{\mathcal{A}^{H_n},\mathrm{\Pi}^{H_n}}(n)=1\bigm|H_n=h_n\right]\le \gamma+\frac{1}{n^d}.
\end{equation}
The specific sequence $\{h_n\}$ of functions is an \emph{instantiation} of the random oracle $\{H_n\}$. 
Note that, in the case where $\{h_n\}$ is polynomial-time computable, i.e., there exists a deterministic Turing machine which on every input $(1^n,x)$
operates and outputs
$h_n(x)$ within time polynomial in $n$, the condition \eqref{gfsrom-rel} implies that the scheme $\Pi$ is just \emph{secure in the standard model}.
Here, the standard model is the normal model of a cryptographic scheme, where no random oracle is present.

In this paper,
we investigate the properties of a specific oracle $\{h_n\}$ relative to which $\Pi$ is secure,
under the assumption that
$\Pi$ is secure in the random oracle model.
The contributions of the paper to the random oracle methodology are as follows:
\begin{enumerate}
\item We investigate the instantiation of the random oracle by a random real in a signature scheme already proved secure in the random oracle model.
  We present equivalent conditions for a specific oracle $\{h_n\}$ instantiating the random oracle to keep a
  signature
  scheme secure,
  using a concept of algorithmic randomness, i.e., a variant of Martin-L\"of randomness.
  Based on this, in particular we show that the security proved in the random oracle model is firmly maintained after instantiating the random oracle by a random real.
\item We introduce the notion of \emph{effective security}, which is a constructive strengthen of the conventional (non-constructive) notions of security.
  In terms of the definitions \eqref{gfsrom} and \eqref{gfsrom-rel} of security, the ``effectiveness'' means that 
  the natural number $N$ can be computed from the code of an adversary $\mathcal{A}$ and a natural number $d$.
  We consider signature schemes in the random oracle model, and show that some specific \emph{computable} function $\{h_n\}$ can instantiate the random oracle
  while keeping the effective security originally proved in the random oracle model.
  We demonstrate that the effective security notions are a natural alternative to the conventional security notions in modern cryptography
  by reconsidering the security notions required in modern cryptography.
\end{enumerate}
The results
in the contributions (i) and (ii) above
are based only on the general form of the definitions of security notions for a signature scheme in modern cryptography,
and depend neither on specific schemes nor on specific security notions. 
On the other hand,
our results on the secure instantiation of the random oracle are valid only if the security in the random oracle model is confirmed already.
This may imply that the random oracle model is not necessarily an imaginary framework to discuss
the security of a cryptographic scheme,
but may have substantial implications for the security
in the standard model.

In addition to the random oracle model, in modern cryptography the \emph{generic group model} \cite{S97} is used also as an imaginary framework
in which the security of a cryptographic scheme is discussed.
In particular, the generic group model is often used to discuss the \emph{computational hardness} of problems, such as the discrete logarithm problem
and the Diffie-Hellman problem in finite cyclic groups,
which is used as a computational hardness assumption to prove the security of a cryptographic scheme.
In the generic group model, the \emph{generic group}, i.e., a random encoding of the group elements, is an imaginary object, just like the random oracle.
Therefore, even if the security of a cryptographic scheme or the hardness of a computational problem is proved in the generic group model,
the generic group has to be instantiated using a concrete finite cyclic group whose group operations are efficiently computable,
for use of the cryptographic scheme in the real world.
Hence, the problem of the secure instantiation of the generic group
exists
in the generic group model, just like in the random oracle model.

In this paper we introduce the notion of \emph{effective hardness} for computational problems,
which corresponds to the effective security for cryptographic schemes.
Based on concepts and methods of algorithmic randomness, we then show that,
for the discrete logarithm problem and the Diffie-Hellman problem in the generic group model,
the generic group can be instantiated by a specific computable function while keeping the effective hardness originally proved in the generic group model.
This result corresponds to the contribution (ii) above for the random oracle model.
We can show the results for the generic group model which corresponds to the contribution (i) for the random oracle model. 
However, this task is not difficult and therefore omitted in this paper. 

\subsection{Organization of the paper}

The paper is organized as follows.
As preliminaries we first review some definitions and results of algorithmic randomness in Section~\ref{preliminaries}.

We then begin the study of the random oracle model in Section~\ref{CSIROSS},
where we present the general form of signature schemes which we consider in this paper,
and introduce the (conventional) security notion for the signature schemes.
In Section~\ref{CSIRO} we present the contribution (i) above for the random oracle model.
Subsequently we present the contribution (ii) above in Sections~\ref{computable-function} and \ref{discussion},
where we introduce the notion of effective security and
then show a secure instantiation of the random oracle by a computable function in Sections~\ref{computable-function},
and we demonstrate the importance of the effective security notions in Section~\ref{discussion}.

We then begin the study of the generic group model in Section~\ref{DL}, where we explain the discrete logarithm problem in the generic group model.
In Section~\ref{LM} we develop the Lebesgue outer measure on families of encoding functions.
It is needed in Section~\ref{DL2},
where we introduce the notion of effective hardness and
then show
a secure instantiation of the generic group in the discrete logarithm problem by a computable function.
In Section~\ref{CDH} we show that the same results hold for the Diffie-Hellman problem.
We conclude this paper with the clarification of
the notion of effective hardness in Section~\ref{conclusion-ggm}.

\section{Preliminaries}
\label{preliminaries}

We start with some notation about numbers and strings which will be used in this paper.
$\#S$ is the cardinality of $S$ for any set $S$.
$\N=\left\{0,1,2,3,\dotsc\right\}$ is the set of natural numbers, and $\N^+$ is the set of positive integers.
$\Q$ is the set of rationals, and $\R$ is the set of reals.

$\X=
\left\{
  \lambda,0,1,00,01,10,11,000,001,010,\dotsc
\right\}$
is the set of finite binary strings where $\lambda$ denotes the \emph{empty string}, and $\X$ is ordered as indicated.
We identify any string in $\X$ with a natural number in this order,
i.e., we consider a map $\varphi\colon \X\to\N$ such that $\varphi(x)=1x-1$ where
the concatenation $1x$ of the strings $1$ and $x$ is regarded as a dyadic integer, and then we identify $x$ with $\varphi(x)$.
For any $x \in \X$, $\abs{x}$ is the \emph{length} of $x$.
For any $n\in\N$,
we denote by $\{0,1\}^n$ and $\{0,1\}^{\le n}$ the sets $\{\,x\mid x\in\X\;\&\;\abs{x}=n\}$ and $\{\,x\mid x\in\X\;\&\;\abs{x}\le n\}$, respectively.
For any $n,m\in\N$, we denote by $\mathsf{Func}_n^m$ and $\mathsf{Func}_{\le n}^m$
the set of all functions mapping $\{0,1\}^n$ to $\{0,1\}^m$ and the set of all functions mapping $\{0,1\}^{\le n}$ to $\{0,1\}^m$, respectively.
A subset $S$ of $\X$ is called
\emph{prefix-free}
if no string in $S$ is a prefix of another string in $S$.
We write ``r.e.'' instead of ``recursively enumerable.''

$\XI$ is the set of infinite binary sequences, where an infinite binary sequence is infinite to the right but finite to the left.
For any $\alpha\in\XI$ and any $n\in\N$, we denote by $\rest{\alpha}{n}\in\X$ the first $n$ bits of $\alpha$.
For any $S\subset\X$, the set $\{\alpha\in\XI\mid\exists\,n\in\N\;\rest{\alpha}{n}\in S\}$ is denoted by $\osg{S}$.
Note that (i) $\osg{S}\subset\osg{T}$ for every $S\subset T\subset\X$,
and (ii) for every set $S\subset\X$ there exists a prefix-free set $P\subset\X$ such that $\osg{S}=\osg{P}$.

Lebesgue outer measure $\mathcal{L}$ on $\XI$ is a function mapping any subset of $\XI$ to a non-negative real.
In this paper, we use the following properties of $\mathcal{L}$.

\begin{proposition}[Properties of Lebesgue outer measure on $\XI$]\label{Lebesgue-outer}\hfill
\begin{enumerate}
\item $\Lm{\osg{P}}=\sum_{x\in P}2^{-\abs{x}}$ for every prefix-free set $P\subset\X$.
  Therefore $\Lm{\emptyset}=\Lm{\osg{\emptyset}}=0$ and $\Lm{\XI}=\Lm{\osg{\{\lambda\}}}=1$.
\item $\Lm{\mathcal{C}}\le\Lm{\mathcal{D}}$ for every $\mathcal{C}\subset\mathcal{D}\subset\XI$.
\item $\Lm{\bigcup_{i}\mathcal{C}_i}\le \sum_{i}\Lm{\mathcal{C}_i}$ for every sequence $\{\mathcal{C}_i\}_{i\in\N}$ of subsets of $\XI$.\qed
\end{enumerate}
\end{proposition}

A function $f\colon\N\to\X$ or $f\colon\N\to\Q$ is called \emph{computable}
if there exists a deterministic Turing machine which on every input $n\in\N$ halts and outputs $f(n)$.
A computable function is also called a \emph{total recursive function}.
A real $a$ is called \textit{computable} if there exists a computable function $g\colon\N\to\Q$ such that $\abs{a-g(k)} < 2^{-k}$ for all $k\in\N$.
We say that $\alpha\in\XI$ is \emph{computable} if the mapping $\N\ni n\mapsto\rest{\alpha}{n}$ is a computable function,
which is equivalent to that the real $0.\alpha$ in base-two notation is computable.
\subsection{Algorithmic randomness}
\label{AR}

In the following we concisely review some definitions and results of algorithmic randomness 
\cite{C75,C87b,N09,DH10}.
The idea in algorithmic randomness is to think of a real, i.e., an infinite binary sequence, as random if it is in no \emph{effective null set}.
An effective null set is a subset $\mathcal{S}$ of $\XI$ such that $\Lm{\mathcal{S}}=0$ and $\mathcal{S}$
has some type of effective property.
To specify an algorithmic randomness notion, one has to specify a type of effective null set, which is usually done by introducing a test concept.
Failing the test is the same as being in the null set.
In this manner, various randomness notions,
such as $2$-randomness, weak $2$-randomness, Demuth randomness, Martin-L\"of randomness, Schnorr randomness, Kurtz randomness,
have been introduced so far, and a hierarchy of algorithmic randomness notions has been developed
(see \cite{N09,DH10} for the detail).

Among all randomness notions, \emph{Martin-L\"of randomness} is a central one.
This is because in many respects, Martin-L\"of randomness is well-behaved,
in that the many properties of Martin-L\"of random infinite sequences do match our intuition of what random infinite sequence should look like.
Moreover, the concept of Martin-L\"of randomness is robust in the sense that it admits various equivalent definitions that are all natural and intuitively meaningful,
as we will see
in Theorem~\ref{equivMLR}.
Martin-L\"of randomness is defined as follows based on the notion of \emph{Martin-L\"of test}.

\begin{definition}[Martin-L\"{o}f randomness, Martin-L\"{o}f \cite{M66}]\label{ML-randomness}
A subset $\mathcal{C}$ of $\N^+\times\X$ is called a \emph{Martin-L\"{o}f test} if $\mathcal{C}$ is an r.e.~set and
there exists a total recursive function $f\colon\N^+\to\Q\cap(0,\infty)$ such that
$\lim_{n\to\infty}f(n)=0$
and for every $n\in\N^+$,
\begin{equation*}
  \Lm{\osg{\mathcal{C}_n}}\le f(n),
\end{equation*}
where
$\mathcal{C}_n
=
\left\{\,
  x\bigm|(n,x)\in\mathcal{C}
\,\right\}$.

For any $\alpha\in\XI$, we say that $\alpha$ is \emph{Martin-L\"{o}f random} if
for every Martin-L\"{o}f test $\mathcal{C}$
there exists $n\in\N^+$ such that $\alpha\notin\osg{\mathcal{C}_n}$.%
\footnote{%
Normally, Martin-L\"of random is defined with fixing the total recursive function $f\colon\N^+\to\Q\cap(0,\infty)$ to the form $f(n)=2^{-n}$.
However, the relaxation of the function $f$ as in Definition~\ref{ML-randomness} does not alter the class of Martin-L\"of random infinite binary sequences.}
\qed
\end{definition}

Let $\mathcal{C}$ be a Martin-L\"{o}f test.
Then, for each $k\in\N^+$, using (ii) of Proposition~\ref{Lebesgue-outer} we see that
$\Lm{\bigcap_{n=1}^{\infty}\osg{\mathcal{C}_n}}\le\Lm{\osg{\mathcal{C}_k}}\le f(k)$.
On letting $k\to\infty$, we have $\Lm{\bigcap_{n=1}^{\infty}\osg{\mathcal{C}_n}}=0$.
Thus, the set $\bigcap_{n=1}^{\infty}\osg{\mathcal{C}_n}$ forms an effective null set in the notion of Martin-L\"{o}f randomness.
Definition~\ref{ML-randomness} says that an infinite binary sequence $\alpha$ is Martin-L\"of random
if $\alpha$ is not in the effective null set $\bigcap_{n=1}^{\infty}\osg{\mathcal{C}_n}$ for any Martin-L\"{o}f test $\mathcal{C}$.

One of the equivalent variants
of Martin-L\"of randomness is Solovay randomness, which plays
an important
role in this paper, as well as Martin-L\"of randomness.

\begin{definition}[Solovay randomness, Solovay~\cite{Sol75}]
A subset $\mathcal{C}$ of $\N^+\times\X$ is called a \emph{Solovay test} if $\mathcal{C}$ is an r.e.~set and
\begin{equation*}
  \sum_{n=1}^{\infty} \Lm{\osg{\mathcal{C}_n}}<\infty.
\end{equation*}

For any $\alpha\in\XI$, we say that $\alpha$ is \emph{Solovay random} if for every Solovay test $\mathcal{C}$,
there exists $N\in\N^+$ such that,
for every $n\ge N$, $\alpha\notin\osg{\mathcal{C}_n}$.
\qed
\end{definition}

For each Solovay test $\mathcal{C}$, we can show that $\Lm{\bigcap_{n=1}^{\infty}\bigcup_{k=n}^{\infty}\osg{\mathcal{C}_k}}=0$.
The set $\bigcap_{n=1}^{\infty}\bigcup_{k=n}^{\infty}\osg{\mathcal{C}_k}$ forms an effective null set in the notion of Solovay randomness.

The robustness of Martin-L\"of randomness is mainly due to the fact that it admits
characterizations based on the notion of program-size complexity,
as shown in Theorem~\ref{equivMLR}.
The \emph{program-size complexity} (or \emph{Kolmogorov complexity}) $K(x)$ of a finite binary string $x$
is defined as the length of the shortest binary input
for a universal decoding algorithm $U$,
called an \textit{optimal prefix-free machine},
to output $x$ (see Chaitin~\cite{C75} for the detail).
By the definition,
$K(x)$ can be thought of as
the randomness contained in the individual finite binary string $x$.

\begin{theorem}[Schnorr~\cite{Sch73}, Chaitin~\cite{C75}, and Solovay~\cite{Sol75}, and Miller and Yu~\cite{MY08}]\label{equivMLR}
For every $\alpha\in\XI$, the following conditions are equivalent:
\begin{enumerate}
  \item $\alpha$ is Martin-L\"{o}f random.
  \item $\alpha$ is Solovay random.
  \item There exists $c\in\N$ such that, for all $n\in\N^+$, $n-c \le K(\rest{\alpha}{n})$.
  \item $\sum_{n=1}^{\infty}2^{n-K(\reste{\alpha}{n})}<\infty$.\qed
\end{enumerate}
\end{theorem}

In particular, the condition (iii) means that the infinite binary sequence $\alpha$ is incompressible.

We denote by $\mathsf{MLR}$ the set of all infinite binary sequences which are Martin-L\"of random.
Since there are only countably infinitely many algorithms and every Martin-L\"of test induces an effective null set,
it is easy to show the following theorem.

\begin{theorem}[Martin-L\"of~\cite{M66}]\label{MLR-measure1}
$\mathcal{L(\mathsf{MLR})}=1$.
\qed
\end{theorem}

\section{Signature schemes and their security}
\label{CSIROSS}

We begin by presenting the general form of signature scheme whose security we consider
in this paper.
For modern cryptography in general,
we refer the reader to
Katz and Lindell~\cite{KL07}.

In 1993 Bellare and Rogaway proposed the notion of \emph{full-domain hash} (FDH) signature scheme
in their original paper on the random oracle model \cite{BR93}.
They showed that
the RSA-FDH signature scheme, which is an instantiation of the FDH signature scheme with the RSA function as a trapdoor permutation,
is effectively existentially unforgeable under an adaptive chosen-message attack (EUF-ACMA secure) in the random oracle model
under the RSA assumption (see Theorem~\ref{RSA-FDH}; for the detail of RSA-FDH see also \cite[Chapter 13]{KL07}).
In the first half of this paper,
we consider a general form of the FDH signature scheme and give our results about the secure instantiation of the random oracle for
that general scheme.

Let $\ell(n)$ be a polynomial with integer coefficients such that $\ell(n)>0$ for all $n\in\N^+$.
An \emph{$\ell$-function} is a function $H\colon\N\times\X\to\X$ such that $\abs{H(n,x)}=\ell(n)$ for all $n\in\N$ and $x\in\X$.
For each $\ell$-function $H$ and $n\in\N$, we define a function $H_n\colon\X\to\{0,1\}^{\ell(n)}$ by $H_n(x)=H(n,x)$. 
An $\ell$-function serves as an instantiation of the random oracle, such as a cryptographic hash function.

\begin{definition}\label{gfFDH}
Let $\ell(n)$ be a polynomial.
A \emph{signature scheme relative to $\ell$-functions} is
a tuple $(\mathsf{Gen},\mathsf{Sign},\mathsf{Vrfy})$ of three polynomial-time algorithms such that, for every $\ell$-function $H$,
\renewcommand{\labelenumi}{\arabic{enumi}.}
\begin{enumerate}
\item The \emph{key generation algorithm} $\mathsf{Gen}$ is a probabilistic algorithm which
  takes as input a security parameter $1^n$ and outputs a pair of keys $(pk,sk)$.
  These are called the \emph{public key} and the \emph{private key}, respectively.
  We assume
  that $n$ can be determined from each of $pk$ and $sk$.
\item The \emph{signing algorithm} $\mathsf{Sign}$ is a probabilistic algorithm which takes as input a private key $sk$ and a \emph{message} $m\in\X$.
  It is given oracle access to $H_n(\cdot)$, and then outputs a \emph{signature} $\sigma$, denoted as $\sigma\gets\mathsf{Sign}_{sk}^{H_n(\cdot)}(m)$.
\item The \emph{verification algorithm} $\mathsf{Vrfy}$ is a deterministic algorithm which
  takes as input a public key $pk$, a massage $m$, and a signature $\sigma$.
  It is given oracle access to $H_n(\cdot)$, and then outputs a bit $b$, with $b=1$ meaning \emph{valid} and $b=0$ meaning \emph{invalid}.
  We write this as $b:=\mathsf{Vrfy}_{pk}^{H_n(\cdot)}(m,\sigma)$.
\end{enumerate}
It is required that, for every $\ell$-function $H$, for every $n\in\N^+$, for every $(pk,sk)$ output by $\mathsf{Gen}(1^n)$, and for every $m\in\X$,
\begin{equation}\label{VS=1}
  \mathsf{Vrfy}_{pk}^{H_n(\cdot)}(m,\mathsf{Sign}_{sk}^{H_n(\cdot)}(m))=1.
\end{equation}
\qed
\end{definition}

In general, a signature scheme is used in the following way.
One party $S$, who acts as the \emph{signer}, runs $\mathsf{Gen}(1^n)$ to obtain keys $(pk,sk)$.
The public key $pk$ is then publicized as belonging to $S$; e.g., $S$ can put the public key on its webpage or place it in some public directory.
We assume that any other party is able to obtain a legitimate copy of $S$'s public key.
When $S$ wants to transmit a message $m$, it computes $\sigma\gets\mathsf{Sign}_{sk}^{H_n(\cdot)}(m)$ and sends $(m, \sigma)$.
Upon receipt of $(m,\sigma)$, a receiver who knows $pk$ can verify the authenticity of $m$ by checking whether
$\mathsf{Vrfy}_{pk}^{H_n(\cdot)}(m,\sigma)=1$, or not.
This establishes both that $S$ sent $m$, and also that $m$ was not modified in transmit.
Note here that Definition~\ref{gfFDH} only defines the syntax of signature schemes and does not define the security of them at all,
which is defined in what follows.

As the security notion of signature schemes,
in this paper we consider \emph{the existential unforgeability
under adaptive chosen-message attacks} (\emph{EUF-ACMA security}) as an example.
We can show the same results for other security notions,
such as the existential unforgeability against key only attacks (EUF-KOA security),
the existential unforgeability against known-message attacks (EUF-KMA security),
and the existential unforgeability against generic chosen-massage attacks  (EUF-GCMA security),
which are all weaker than the EUF-ACMA security.

Given a public key $pk$ generated by a signer $S$ to an adversary, we say that the adversary outputs a \emph{forgery}
if it outputs a message $m$ along with a valid signature $\sigma$ on $m$,
and furthermore $m$ was not previously signed by $S$ using the private key $sk$ which corresponds to $pk$.
The EUF-ACMA security of a signature scheme means that
an adversary cannot output a forgery even if it is allowed to obtain signatures on many other messages of its choice.
The formal definition is given as follows.

Let $\mathrm{\Pi}=(\mathsf{Gen},\mathsf{Sign},\mathsf{Vrfy})$ be a signature scheme relative to $\ell$-functions,
and consider the following experiment for a probabilistic polynomial-time adversary $\mathcal{A}$,%
\footnote{Normally, a probabilistic (uniform) polynomial-time Turing machine is called a \emph{probabilistic polynomial-time adversary}
when it is used as an adversary against a cryptographic scheme.}
a parameter $n$, and a function $G$ mapping a superset of $\{0,1\}^{\le q(n)}$ to $\{0,1\}^{\ell(n)}$ where $q(n)$ is the maximum value
among the running time of $\mathsf{Sign}$, the running time of $\mathsf{Vrfy}$, and the running time of $\mathcal{A}$ on the parameter $n$:

\begin{quote}
\textbf{The signature experiment $\mathsf{Sig\text{-}forge}_{\mathcal{A},\mathrm{\Pi}}(n,G)$:}
\vspace*{-1mm}
\renewcommand{\labelenumi}{\arabic{enumi}.}
\emph{
\begin{enumerate}
\item $\mathsf{Gen}(1^n)$ is run to obtain keys $(pk,sk)$.
\item Adversary $\mathcal{A}$ is given $pk$ and oracle access to both $\mathsf{Sign}_{sk}^{G(\cdot)}(\cdot)$ and $G(\cdot)$.
  (The first oracle returns a signature $\mathsf{Sign}_{sk}^{G(\cdot)}(m')$ for any message $m'$ of the adversary's choice
  while having oracle access to $G(\cdot)$ of itself.)
  The adversary then outputs $(m,\sigma)$.
  Let $\mathcal{Q}$ denotes the set of messages whose signatures were requested by $\mathcal{A}$ during its execution.
\item The output of the experiment is defined to be $1$
  if both $m\notin\mathcal{Q}$ and $\mathsf{Vrfy}_{pk}^{G(\cdot)}(m,\sigma)=1$ hold true, and $0$ otherwise.
\end{enumerate}
}
\end{quote}

Here the function $G$ serves as an instantiation of the random oracle.
Since the running time of each of $\mathsf{Sign}$, $\mathsf{Vrfy}$, and $\mathcal{A}$ on the parameter $n$ is at most $q(n)$,
the lengths of the strings queried to the oracle $G(\cdot)$ by these three algorithms during their computations are at most $q(n)$.
Thus the function $G$ only have to be defined on the set $\{0,1\}^{\le q(n)}$.

On the one hand,
the EUF-ACMA security of signature schemes
relative to a specific $\ell$-function is defined as follows.
This form of the definition
corresponds to the condition \eqref{gfsrom-rel} with $\gamma=0$ for the security of a signature scheme
relative to a specific oracle $\{h_n\}$ considered in the introduction.

\begin{definition}\label{EUF-ACMA-secure-rel}
Let $H$ be an $\ell$-function. 
A signature scheme $\mathrm{\Pi}=(\mathsf{Gen},\mathsf{Sign},\mathsf{Vrfy})$ relative to $\ell$-functions is
\emph{existentially unforgeable under an adaptive chosen-message attack} (or \emph{EUF-ACMA secure}) \emph{relative to $H$}
if for all probabilistic polynomial-time adversaries $\mathcal{A}$ and all $d\in\N^+$ there exists $N\in\N^+$ such that, for all $n\ge N$,
\begin{equation*}
  \Prob[\mathsf{Sig\text{-}forge}_{\mathcal{A},\mathrm{\Pi}}(n,H_n)=1]\le\frac{1}{n^d}.
\end{equation*}
\qed
\end{definition}

On the other hand,
the EUF-ACMA security of signature schemes
in the random oracle model is formulated as follows.
This form of the definition corresponds to the condition \eqref{rw-gfsrom}, and is justified based on the consideration in the introduction.

\begin{definition}\label{EUF-ACMA-secure-ro}
A signature scheme $\mathrm{\Pi}=(\mathsf{Gen},\mathsf{Sign},\mathsf{Vrfy})$ relative to $\ell$-functions is
\emph{existentially unforgeable under an adaptive chosen-message attack} (or \emph{EUF-ACMA secure}) \emph{in the random oracle model}
if for all probabilistic polynomial-time adversaries $\mathcal{A}$ and all $d\in\N^+$ there exists $N\in\N^+$ such that, for all $n\ge N$,
\begin{equation*}
  \frac{1}{\#\mathsf{Func}_{\le q(n)}^{\ell(n)}}\sum_{G\in\mathsf{Func}_{\le q(n)}^{\ell(n)}}
  \Prob[\mathsf{Sig\text{-}forge}_{\mathcal{A},\mathrm{\Pi}}(n,G)=1]\le\frac{1}{n^d},
\end{equation*}
where $q(n)$ is the maximum value among
the running time of $\mathsf{Sign}$, the running time of $\mathsf{Vrfy}$, and the running time of $\mathcal{A}$ on the parameter $n$.
\qed
\end{definition}

\section{Conditions for secure instantiation of the random oracle}
\label{CSIRO}

In this section, we present \emph{equivalent} conditions for a specific oracle instantiating the random oracle to keep a signature scheme secure,
using a concept of algorithmic randomness.

In order to apply the method of algorithmic randomness to the random oracle methodology,
we identify an $\ell$-function with an infinite binary sequence in the following manner:
We first choose a particular bijective total recursive function $b\colon\N\to\N\times\N$ with $b(k)=(b_1(k),b_2(k))$
as the standard one for use throughout the rest of this paper.
We assume for convenience that, for every $k,l\in\N$, if $b_1(k)=b_1(l)$ and $k<l$ then $b_2(k)<b_2(l)$.
For example, the inverse function of a function $c\colon\N\times\N\to\N$ with $c(m,n)=(m+n)(m+n+1)/2+n$ can serve as such a function $b$.
Then each $\ell$-function $H\colon\N\times\X\to\X$ is identified with the infinite binary sequence
\begin{equation}\label{identifylf}
  H(b(0))H(b(1))H(b(2))H(b(3))\dotsm\dotsm,
\end{equation}
where the countably infinite finite binary strings $H(b(0)),H(b(1)),H(b(2)),H(b(3)),\dotsc$ are concatenated.
Recall
that we identify $\X$ with $\N$, as explained in Section~\ref{preliminaries},
and therefore each $b_2(k)$ is regarded as a finite binary string in \eqref{identifylf}.
In what follows, we work with this intuition of the identification.

We will give
the main result of this section,
i.e., Theorem~\ref{signature_SR-MLR-like}, in terms of Solovay randomness and Martin-L\"{o}f randomness.
For that purpose we generalize these two randomness notions in Definitions~\ref{gnrl-S} and \ref{gnrl-ML}, respectively.

\begin{definition}[Solovay randomness with respect to an arbitrary set of Solovay tests]\label{gnrl-S}
Let $S$ be a set of Solovay tests.
For any $\alpha\in\XI$, we say that $\alpha$ is \emph{Solovay random with respect to $S$} if for every Solovay test $\mathcal{C}\in S$,
there exists $N\in\N^+$ such that,
for every $n\ge N$, $\alpha\notin\osg{\mathcal{C}_n}$.
\qed
\end{definition}

\begin{definition}\label{def-stest_acma}
Let $\ell(n)$ be a polynomial, and let $\mathrm{\Pi}=(\mathsf{Gen},\mathsf{Sign},\mathsf{Vrfy})$ be a signature scheme relative to $\ell$-functions.

For each probabilistic polynomial-time adversary $\mathcal{A}$ and each $d,n\in\N^+$
we define a subset $\osg{C_{\mathcal{A},d,n}}$ of $\XI$ as the set of all $\ell$-functions $H$ such that
\begin{align*}
  \Prob[\mathsf{Sig\text{-}forge}_{\mathcal{A},\mathrm{\Pi}}(n,H_n)=1]>\frac{1}{n^d}.
\end{align*}
To be precise,
we define a subset $C_{\mathcal{A},d,n}$ of $\X$ as the set of all finite binary strings of the form
$x_0G(\lambda)x_1G(0)x_2G(1)x_3\dotsm x_L G(1^{q})$
such that the following properties (i), (ii), (iii), and (iv) hold for $q$, $L$, $x_0, x_1, x_2, x_3, \dots, x_L$, and $G$:
\begin{enumerate}
\item $q$ is the maximum value among the running time of $\mathsf{Sign}$, the running time of $\mathsf{Vrfy}$,
  and the running time of $\mathcal{A}$ on the parameter $n$.
\item $L+1=\#\{0,1\}^{\le q}$ (i.e., $L=2^{q+1}-2$).
\item %
  For each $j\in\{0,\dots,L\}$,
  $x_j\in\X$ and
  \begin{equation*}
    \abs{x_0G(\lambda)x_1G(0)x_2G(1)x_3\dotsm x_{j}}=\sum_{k=0}^{k_j-1}\ell(b_1(k))
  \end{equation*}
  where $k_j$ is
  a natural number such that
  $b(k_j)=(n,j)$.
\item %
  $G\in\mathsf{Func}_{\le q}^{\ell(n)}$and
  \begin{equation*}
    \Prob[\mathsf{Sig\text{-}forge}_{\mathcal{A},\mathrm{\Pi}}(n,G)=1]>\frac{1}{n^d}.
  \end{equation*}
\end{enumerate}

We then define $\mathsf{S\text{-}TEST}_{\mathrm{\Pi}}^{\mathsf{EUF\text{-}ACMA}}$ as the class of all subsets $\mathcal{C}$ of $\N^+\times\X$
for which there exist a probabilistic polynomial-time adversary $\mathcal{A}$ and $d\ge 2$
such that $\mathcal{C}=\{(n,y)\mid n\in\N^+\text{ \& }y\in C_{\mathcal{A},d,n}\}$.
\qed
\end{definition}

\begin{theorem}\label{main-Solovay-signature}
Let $\ell(n)$ be a polynomial.
Suppose that a signature scheme $\mathrm{\Pi}=(\mathsf{Gen},\mathsf{Sign},\mathsf{Vrfy})$ relative to $\ell$-functions is
EUF-ACMA secure in the random oracle model.
Then $\mathsf{S\text{-}TEST}_{\mathrm{\Pi}}^{\mathsf{EUF\text{-}ACMA}}$ contains only Solovay tests.\qed
\end{theorem}

In order to prove Theorem~\ref{main-Solovay-signature}, we need the following two lemmas.

\begin{lemma}\label{lemma1}
Let $f_1,\dots,f_N$ be reals.
Suppose that
$\frac{1}{N}\sum_{i=1}^{N} f_i \le \varepsilon$.
Then, for every $\alpha>0$, the number of $i$ for which $\alpha\varepsilon< f_i$ is less than $N/\alpha$.
\end{lemma}

\begin{proof}%
We prove the contraposition of Lemma~\ref{lemma1}. 
Assume that the number of $i$ for which $\alpha\varepsilon< f_i$ is at least $N/\alpha$.
Then $\sum_{i=1}^{N} f_i>\alpha\varepsilon N/\alpha=\varepsilon N$ and therefore $\frac{1}{N}\sum_{i=1}^{N} f_i>\varepsilon$.
\end{proof}

\begin{lemma}\label{lemma2}
Let $d\ge 2$. Then $\sum_{k=n}^{\infty}1/k^d\le 2/n$ for every $n\in\N^+$.
\end{lemma}

\begin{proof}%
In the case of $n\ge 2$, we
have
\begin{equation}\label{in_case_of_nle2}
\begin{split}
  \sum_{k=n}^{\infty}\frac{1}{k^d}\le\sum_{k=n}^{\infty}\int_{k-1}^{k}\frac{1}{k^d}=\int_{n-1}^{\infty} \frac{1}{x^d}dx=\frac{1}{(d-1)(n-1)^{d-1}}
  \le \frac{1}{n-1}\le \frac{1}{n-n/2}=\frac{2}{n}.
\end{split}
\end{equation}
On the other hand, in the case of $n=1$, using \eqref{in_case_of_nle2} we
have
\begin{align*}
  \sum_{k=n}^{\infty}\frac{1}{k^d}=1+\sum_{k=2}^{\infty}\frac{1}{k^d}\le 1+\frac{2}{2}=2=\frac{2}{n}.
\end{align*}
Thus $\sum_{k=n}^{\infty}1/k^d\le 2/n$ holds in any case.
\end{proof}

\begin{proof}[Proof of Theorem~\ref{main-Solovay-signature}]
Let $\mathcal{C}\in\mathsf{S\text{-}TEST}_{\mathrm{\Pi}}^{\mathsf{EUF\text{-}ACMA}}$.
Then there exist a probabilistic polynomial-time adversary $\mathcal{A}$ and
$d\ge 2$
such that, for every $n\in\N^+$, $\mathcal{C}_n=C_{\mathcal{A},d,n}$.
Suppose that $\mathrm{\Pi}$ is EUF-ACMA secure in the random oracle model.
Then it follows from Definition~\ref{EUF-ACMA-secure-ro} that there exists $N\in\N^+$ such that, for all $n\ge N$,
\begin{equation}\label{03310330}
  \frac{1}{\#\mathsf{Func}_{\le q(n)}^{\ell(n)}}\sum_{G\in\mathsf{Func}_{\le q(n)}^{\ell(n)}}
  \Prob[\mathsf{Sig\text{-}forge}_{\mathcal{A},\mathrm{\Pi}}(n,G)=1]\le\frac{1}{n^{2d}},
\end{equation}
where $q(n)$ is the maximum value among
the running time of $\mathsf{Sign}$, the running time of $\mathsf{Vrfy}$, and the running time of $\mathcal{A}$ on the parameter $n$.

On the one hand, it follows from Definition~\ref{def-stest_acma} that $\mathcal{C}$ is an r.e.~set,
since the dyadic rational $$\Prob[\mathsf{Sig\text{-}forge}_{\mathcal{A},\mathrm{\Pi}}(n,G)=1]$$ is computable,
given $n$ and $G\in\mathsf{Func}_{\le q(n)}^{\ell(n)}$.

On the other hand, using \eqref{03310330} and Lemma~\ref{lemma1} with $\varepsilon=1/n^{2d}$ and $\alpha=n^d$, we see that, for every $n\ge N$,
\begin{align*}
  \#\left\{\,G\in\mathsf{Func}_{\le q(n)}^{\ell(n)}\biggm|\Prob[\mathsf{Sig\text{-}forge}_{\mathcal{A},\mathrm{\Pi}}(n,G)=1]>\frac{1}{n^d}\,\right\}
  <\frac{\#\mathsf{Func}_{\le q(n)}^{\ell(n)}}{n^d}.
\end{align*}
Since $$\#\mathsf{Func}_{\le q(n)}^{\ell(n)}=2^{\ell(n)\#\{0,1\}^{\le q(n)}},$$
it follows from Definition~\ref{def-stest_acma} and (i) of Proposition~\ref{Lebesgue-outer} that
\begin{align*}
  \sum_{n=N}^{\infty}\Lm{\osg{\mathcal{C}_n}}=\sum_{n=N}^{\infty}\sum_{y\in\mathcal{C}_n}2^{-\abs{y}}
  <\sum_{n=N}^{\infty}\frac{\#\mathsf{Func}_{\le q(n)}^{\ell(n)}}{n^d}2^{-\ell(n)\#\{0,1\}^{\le q(n)}}
  =\sum_{n=N}^{\infty}\frac{1}{n^d}<\infty,
\end{align*}
where the last inequality follows from Lemma~\ref{lemma2}.
Thus $\mathcal{C}$ is a Solovay test.
\end{proof}

\begin{definition}[Martin-L\"{o}f randomness with respect to an arbitrary set of Martin-L\"{o}f tests]\label{gnrl-ML}
Let $S$ be a set of Martin-L\"{o}f tests.
For any $\alpha\in\XI$, we say that $\alpha$ is \emph{Martin-L\"{o}f random with respect to $S$} if
for every Martin-L\"{o}f test $\mathcal{C}\in S$, there exists $n\in\N^+$
such that $\alpha\notin\osg{\mathcal{C}_n}$.
\qed
\end{definition}

\begin{definition}
Let $\ell(n)$ be a polynomial, and let $\mathrm{\Pi}=(\mathsf{Gen},\mathsf{Sign},\mathsf{Vrfy})$ be a signature scheme relative to $\ell$-functions.
We define $\mathsf{ML\text{-}TEST}_{\mathrm{\Pi}}^{\mathsf{EUF\text{-}ACMA}}$ as the class of all subsets $\mathcal{C}$ of $\N^+\times\X$
for which there exists $\mathcal{D}\in\mathsf{S\text{-}TEST}_{\mathrm{\Pi}}^{\mathsf{EUF\text{-}ACMA}}$ such that, for every $n\in\N^+$,
$\mathcal{C}_n=\bigcup_{k=n}^\infty\mathcal{D}_k$.
\qed
\end{definition}

\begin{theorem}\label{main-ML-signature}
Let $\ell(n)$ be a polynomial.
Suppose that a signature scheme $\mathrm{\Pi}=(\mathsf{Gen},\mathsf{Sign},\mathsf{Vrfy})$ relative to $\ell$-functions is
EUF-ACMA secure in the random oracle model.
Then $\mathsf{ML\text{-}TEST}_{\mathrm{\Pi}}^{\mathsf{EUF\text{-}ACMA}}$ contains only Martin-L\"of tests.%
\footnote{In fact,  $\mathsf{ML\text{-}TEST}_{\mathrm{\Pi}}^{\mathsf{EUF\text{-}ACMA}}$ contains only Schnorr tests,
where a Schnorr test is defined as a Martin-L\"of test $\mathcal{C}\subset\N^+\times\X$ such that $\Lm{\osg{\mathcal{C}_n}}$ is computable uniformly in $n$.
For the detail of Schnorr tests, see e.g.~Section 3.5 of Nies~\cite{N09}.}
\end{theorem}

\begin{proof}
Let $\mathcal{C}\in\mathsf{ML\text{-}TEST}_{\mathrm{\Pi}}^{\mathsf{EUF\text{-}ACMA}}$.
Then there exists $\mathcal{D}\in\mathsf{S\text{-}TEST}_{\mathrm{\Pi}}^{\mathsf{EUF\text{-}ACMA}}$ such that, for every $n\in\N^+$,
$\mathcal{C}_n=\bigcup_{k=n}^\infty\mathcal{D}_k$.
Suppose that $\mathrm{\Pi}$ is EUF-ACMA secure in the random oracle model.
It follows from Theorem~\ref{main-Solovay-signature} that $\mathcal{D}$ is a Solovay test.
It is then easy to see that $\mathcal{C}$ is an r.e.~set, since $\mathcal{D}$ is an r.e.~set.
On the other hand,
since $\mathcal{D}\in\mathsf{S\text{-}TEST}_{\mathrm{\Pi}}^{\mathsf{EUF\text{-}ACMA}}$,
there exist a probabilistic polynomial-time adversary $\mathcal{A}$ and $d\ge 2$ such that, for every $n\in\N^+$, $\mathcal{D}_n=C_{\mathcal{A},d,n}$.
Then, in the same manner as the proof of Theorem~\ref{main-Solovay-signature} we can show that there exists $N\in\N^+$ such that,
for every $n\ge N$,
\begin{align*}
  \#\left\{\,G\in\mathsf{Func}_{\le q(n)}^{\ell(n)}\biggm|\Prob[\mathsf{Sig\text{-}forge}_{\mathcal{A},\mathrm{\Pi}}(n,G)=1]>\frac{1}{n^d}\,\right\}
  <\frac{\#\mathsf{Func}_{\le q(n)}^{\ell(n)}}{n^d},
\end{align*}
where $q(n)$ is the maximum value among
the running time of $\mathsf{Sign}$, the running time of $\mathsf{Vrfy}$, and the running time of $\mathcal{A}$ on the parameter $n$.
It follows from (i) and (iii) of Proposition~\ref{Lebesgue-outer} and Definition~\ref{def-stest_acma} that, for each $n\ge N$,
\begin{equation*}%
\begin{split}
  \Lm{\osg{\mathcal{C}_n}}
  \le\sum_{k=n}^{\infty}\Lm{\osg{\mathcal{D}_k}}
  =\sum_{k=n}^{\infty}\sum_{y\in\mathcal{D}_k}2^{-\abs{y}}
  <\sum_{k=n}^{\infty}\frac{\#\mathsf{Func}_{\le q(k)}^{\ell(k)}}{k^d}2^{-\ell(k)\#\{0,1\}^{\le q(k)}}
  =\sum_{k=n}^\infty\frac{1}{k^d}\le \frac{2}{n},
\end{split}
\end{equation*}
where the last inequality follows from Lemma~\ref{lemma2}.
Thus $\mathcal{C}$ is a Martin-L\"of test.
\end{proof}

Obviously, the following proposition holds.

\begin{proposition}\label{plain->S}
Let $\alpha\in\XI$.
\begin{enumerate}
  \item For every set $S$ of Martin-L\"{o}f tests, if $\alpha$ is Martin-L\"{o}f random then $\alpha$ is Martin-L\"{o}f random with respect to $S$.
  \item For every set $S$ of Solovay tests, if $\alpha$ is Solovay random then $\alpha$ is Solovay random with respect to $S$.\qed
\end{enumerate}
\end{proposition}

The following theorem gives equivalent conditions for a specific oracle instantiating the random oracle
to keep the EUF-ACMA security of a signature scheme originally proved in the random oracle model, in terms of algorithmic randomness.

\begin{theorem}[Main result I]\label{signature_SR-MLR-like}
Let $\ell(n)$ be a polynomial.
Suppose that
a signature scheme $\mathrm{\Pi}=(\mathsf{Gen},\mathsf{Sign},\mathsf{Vrfy})$ relative to $\ell$-functions is
EUF-ACMA secure in the random oracle model.
Let $H$ be an $\ell$-function. Then the following conditions are equivalent:
\begin{enumerate}
\item $\mathrm{\Pi}$ is EUF-ACMA secure relative to $H$.
\item $H$ is Solovay random with respect to $\mathsf{S\text{-}TEST}_{\mathrm{\Pi}}^{\mathsf{EUF\text{-}ACMA}}$.
\item $H$ is Martin-L\"of random with respect to $\mathsf{ML\text{-}TEST}_{\mathrm{\Pi}}^{\mathsf{EUF\text{-}ACMA}}$.
\end{enumerate}
\end{theorem}

\begin{proof}
First we show the equivalence between the conditions (i) and (ii).
The negation of the condition (i) is that there exist a probabilistic polynomial-time adversary $\mathcal{A}$ and $d\ge 2$ such that, for infinitely many $n\in\N^+$,
\begin{equation*}
  \Prob[\mathsf{Sig\text{-}forge}_{\mathcal{A},\mathrm{\Pi}}(n,H_n)=1]>\frac{1}{n^d}.
\end{equation*}
However, from Definition~\ref{def-stest_acma}, it is easy to see that this is equivalent to the condition that
there exists $\mathcal{C}\in\mathsf{S\text{-}TEST}_{\mathrm{\Pi}}^{\mathsf{EUF\text{-}ACMA}}$ such that, for infinitely many $n\in\N^+$,
$H\in\osg{\mathcal{C}_n}$.
This is further equivalent to the condition that $H$ is not Solovay random with respect to $\mathsf{S\text{-}TEST}_{\mathrm{\Pi}}^{\mathsf{EUF\text{-}ACMA}}$,
since $\mathsf{S\text{-}TEST}_{\mathrm{\Pi}}^{\mathsf{EUF\text{-}ACMA}}$ contains only Solovay tests by Theorem~\ref{main-Solovay-signature}.
Thus the conditions (i) and (ii) are equivalent to each other.

Next we show the equivalence between the conditions (ii) and (iii).
Suppose that $\mathcal{C}\in\mathsf{ML\text{-}TEST}_{\mathrm{\Pi}}^{\mathsf{EUF\text{-}ACMA}}$ and
$\mathcal{D}\in\mathsf{S\text{-}TEST}_{\mathrm{\Pi}}^{\mathsf{EUF\text{-}ACMA}}$
satisfy that $\mathcal{C}_n=\bigcup_{k=n}^\infty\mathcal{D}_k$ for all $n\in\N^+$.
Then the condition that
$H\in\osg{\mathcal{C}_n}$ for all $n\in\N^+$
is equivalent to the condition that
$H\in\osg{\mathcal{D}_n}$ for infinitely many $n\in\N^+$.
Note here that $\mathsf{ML\text{-}TEST}_{\mathrm{\Pi}}^{\mathsf{EUF\text{-}ACMA}}$ contains only Martin-L\"of tests by Theorem~\ref{main-ML-signature},
and $\mathsf{S\text{-}TEST}_{\mathrm{\Pi}}^{\mathsf{EUF\text{-}ACMA}}$ contains only Solovay tests by Theorem~\ref{main-Solovay-signature}.
Thus, $H$ is not Martin-L\"of random with respect to $\mathsf{ML\text{-}TEST}_{\mathrm{\Pi}}^{\mathsf{EUF\text{-}ACMA}}$ if and only if
$H$ is not Solovay random with respect to $\mathsf{S\text{-}TEST}_{\mathrm{\Pi}}^{\mathsf{EUF\text{-}ACMA}}$.
This completes the proof.
\end{proof}

As noted in the previous section,
Theorem~\ref{FDH-computable} holds for other security notions for signature schemes, such as the EUF-GCMA security, in place of the EUF-ACMA security.
Thus,
given arbitrary security notion UF and signature scheme $\mathrm{\Pi}$ which is UF secure in the random oracle model,
one can define a variant of Martin-L\"of randomness, i.e., Martin-L\"of randomness with respect to $\mathsf{ML\text{-}TEST}_{\mathrm{\Pi}}^{\mathsf{UF}}$,
which gives a equivalent condition for a specific oracle instantiating the random oracle in $\mathrm{\Pi}$ to keep the UF security.
In this manner, given a security notion and a signature scheme satisfying this security notion in the random oracle model,
one can define an algorithmic randomness notion which is specified by an appropriate type of effective null sets based on these security notion and scheme,
and which corresponds exactly to the secure instantiation of the random oracle with respect to this security notion.

In the next section we show in Theorem~\ref{FDH-computable} that
a signature scheme $\mathrm{\Pi}$ can be EUF-ACMA secure relative to some \emph{computable} $\ell$-function $H$,
in the case where $\mathrm{\Pi}$ satisfies a stronger security notion, called the \emph{effective} EUF-ACMA security, in the random oracle model.
Hence, in such a case,
it follows from Theorem~\ref{signature_SR-MLR-like} that
there exists a \emph{computable} infinite binary sequence $H$
which is Martin-L\"of random with respect to $\mathsf{ML\text{-}TEST}_{\mathrm{\Pi}}^{\mathsf{EUF\text{-}ACMA}}$.

The following theorem shows that
the EUF-ACMA security proved in the random oracle model is firmly maintained after instantiating the random oracle by a random real.

\begin{theorem}\label{main2}
Let $\ell(n)$ be a polynomial.
Suppose that
a signature scheme $\mathrm{\Pi}=(\mathsf{Gen},\mathsf{Sign},\mathsf{Vrfy})$ relative to $\ell$-functions is EUF-ACMA secure in the random oracle model.
For every $\ell$-function $H$, if $H$ is Martin-L\"of random then
$\mathrm{\Pi}$ is EUF-ACMA secure relative to $H$.
\end{theorem}

\begin{proof}
The result follows immediately from (i) of Proposition~\ref{plain->S} and Theorem~\ref{signature_SR-MLR-like}.
\end{proof}

The following theorem shows that
a specific oracle instantiating the random oracle almost surely keeps the EUF-ACMA security of a signature scheme originally proved in the random oracle model.

\begin{theorem}\label{Lebesgue2}
Let $\ell(n)$ be a polynomial.
Suppose that
a signature scheme $\mathrm{\Pi}=(\mathsf{Gen},\mathsf{Sign},\mathsf{Vrfy})$ relative to $\ell$-functions is EUF-ACMA secure in the random oracle model.
Then $\mathcal{L}(\mathsf{EUF}_{\mathrm{\Pi}}^{\mathsf{acma}})=1$,
where $\mathsf{EUF}_{\mathrm{\Pi}}^{\mathsf{acma}}$ is the set of all $\ell$-functions $H$ such that $\mathrm{\Pi}$ is EUF-ACMA secure relative to $H$.
\end{theorem}

\begin{proof}
The result
follows immediately from
Theorem~\ref{MLR-measure1}, Theorem~\ref{main2}, and (i) and (ii) of Proposition~\ref{Lebesgue-outer}.
\end{proof}

Impagliazzo and Rudich~\cite{IR88} showed a similar result to Theorem~\ref{Lebesgue2}
for a one-way permutation and derived the negative result about the existence of a secure secret key agreement protocol.
\section{Secure instantiation of the random oracle by computable function}
\label{computable-function}

Let $H$ be an $\ell$-function.
We say that $H$ is \emph{computable} if there exists a deterministic
Turing machine
which on every input $(n,x)$ halts and outputs $H(n,x)$.
On the other hand, we say that $H$ is \emph{polynomial-time computable}
if there exists a deterministic
Turing machine
which on every input $(1^n,x)$ operates and outputs $H(n,x)$ within time polynomial in $n$ and $\abs{x}$.

Conjecture~\ref{FDH-conjecture1} below means that,
in the case where a signature scheme $\mathrm{\Pi}$ satisfies a certain condition $\mathcal{C}$,
the EUF-ACMA security of $\mathrm{\Pi}$ originally proved in the random oracle model
can be firmly maintained in the standard model after instantiating the random oracle by some polynomial-time computable $\ell$-function
(or some polynomial-time computable family of $\ell$-functions).

\begin{conjecture}\label{FDH-conjecture1}
Let $\ell(n)$ be a polynomial.
Suppose that
a signature scheme $\mathrm{\Pi}=(\mathsf{Gen},\mathsf{Sign},\mathsf{Vrfy})$ relative to $\ell$-functions is EUF-ACMA secure in the random oracle model.
If $\mathrm{\Pi}$ satisfies $\mathcal{C}$,
then there exists a polynomial-time computable $\ell$-function (or a polynomial-time computable family of $\ell$-functions) relative to which
$\mathrm{\Pi}$ is EUF-ACMA secure.
\qed
\end{conjecture}

Note that an appropriate restriction on a signature scheme $\mathrm{\Pi}$, i.e., the condition $\mathcal{C}$ on $\mathrm{\Pi}$,
might be necessary to prove Conjecture~\ref{FDH-conjecture1}, due to the negative results in the secure instantiation of the random oracle
by Canetti, Goldreich, and Halevi~\cite{CGH04},
who show ``contrived'' signature schemes (and encryption schemes) that are secure in the random oracle model but are
demonstrably insecure for \emph{any} concrete instantiation of the random oracle.
At present, however, it would seem very difficult to prove it with identifying an appropriate nontrivial condition $\mathcal{C}$.

The second best thing is to investigate whether Conjecture~\ref{FDH-conjecture2} below holds true or not,
where we consider the instantiation of the random oracle by simply a computable $\ell$-function,
which is not necessarily polynomial-time computable.

\begin{conjecture}\label{FDH-conjecture2}
Let $\ell(n)$ be a polynomial.
Suppose that a signature scheme $\mathrm{\Pi}=(\mathsf{Gen},\mathsf{Sign},\mathsf{Vrfy})$ relative to $\ell$-functions is EUF-ACMA secure in the random oracle model.
Then there exists a computable $\ell$-function $H$ such that $\mathrm{\Pi}$ is EUF-ACMA secure relative to $H$.
\qed
\end{conjecture}

In what follows, we show that an ``effective'' variant of Conjecture~\ref{FDH-conjecture2} holds true.
We introduce the notion of
\emph{effective EUF-ACMA security}, which is a constructive strengthen of the conventional (non-constructive) notions of EUF-ACMA security.
In terms of Definitions~\ref{EUF-ACMA-secure-rel} and~\ref{EUF-ACMA-secure-ro} for the conventional EUF-ACMA security,
the ``effectiveness'' means that the number $N$ in the definitions can be computed, given the code of an adversary $\mathcal{A}$ and a number $d$.
To begin with a formal definition,
we choose a particular recursive enumeration $\mathcal{A}_1,\mathcal{A}_2,\mathcal{A}_3,\dotsc$ of all probabilistic polynomial-time adversaries
as the standard one for use throughout the rest of
this section.
It is easy to show that such an enumeration exists.
In fact,
the $k$th probabilistic polynomial-time adversary $\mathcal{A}_k$ can be chosen as
a probabilistic Turing machine obtained by executing the $k$th probabilistic Turing machine $\mathcal{M}_k$ in at most $n^k+k$ steps,
where $n$ is the length of the input of $\mathcal{M}_k$.

On the one hand, the effective EUF-ACMA security relative to a specific $\ell$-function is defined as follows.

\begin{definition}\label{eff-EUACMA}
Let $H$ be an $\ell$-function. 
A signature scheme $\mathrm{\Pi}=(\mathsf{Gen},\mathsf{Sign},\mathsf{Vrfy})$ relative to $\ell$-functions is
\emph{effectively existentially unforgeable under an adaptive chosen-message attack} (or \emph{effectively EUF-ACMA secure}) \emph{relative to $H$}
if there exists a computable function $f\colon\N^+\times\N^+\to\N^+$ such that, for all $i,d,n\in\N^+$, if $n\ge f(i,d)$ then
\begin{equation*}
  \Prob[\mathsf{Sig\text{-}forge}_{\mathcal{A}_i,\mathrm{\Pi}}(n,H_n)=1]\le\frac{1}{n^d}.
\end{equation*}
\qed
\end{definition}

Obviously, if a signature scheme $\mathrm{\Pi}$ relative to $\ell$-functions is effectively EUF-ACMA secure relative to $H$,
then $\mathrm{\Pi}$ is simply EUF-ACMA secure relative to $H$.

On the other hand, the effective EUF-ACMA security in the random oracle model is defined as follows.

\begin{definition}\label{eff-EUACMA-ro}
A signature scheme $\mathrm{\Pi}=(\mathsf{Gen},\mathsf{Sign},\mathsf{Vrfy})$ relative to $\ell$-functions is
\emph{effectively existentially unforgeable under an adaptive chosen-message attack} (or \emph{effectively EUF-ACMA secure}) \emph{in the random oracle model}
if there exists a computable function $f\colon\N^+\times\N^+\to\N^+$ such that, for all $i,d,n\in\N^+$, if $n\ge f(i,d)$ then
\begin{equation*}
  \frac{1}{\#\mathsf{Func}_{\le q_i(n)}^{\ell(n)}}\sum_{G\in\mathsf{Func}_{\le q_i(n)}^{\ell(n)}}
  \Prob[\mathsf{Sig\text{-}forge}_{\mathcal{A}_i,\mathrm{\Pi}}(n,G)=1]\le\frac{1}{n^d},
\end{equation*}
where $q_i(n)$ is the maximum value among the running time of $\mathsf{Sign}$, the running time of $\mathsf{Vrfy}$,
and the running time of $\mathcal{A}_i$ on the parameter $n$.
\qed
\end{definition}

Obviously, if a signature scheme $\mathrm{\Pi}$ relative to $\ell$-functions is effectively EUF-ACMA secure in the random oracle model,
then $\mathrm{\Pi}$ is simply EUF-ACMA secure in the random oracle model.

The effective variant of Conjecture~\ref{FDH-conjecture2} is then presented as follows.

\begin{theorem}[Main result II]\label{FDH-computable}
Let $\ell(n)$ be a polynomial.
Suppose that a signature scheme $\mathrm{\Pi}=(\mathsf{Gen},\mathsf{Sign},\mathsf{Vrfy})$ relative to $\ell$-functions is
effectively EUF-ACMA secure in the random oracle model.
Then there exists a computable $\ell$-function $H$ such that $\mathrm{\Pi}$ is effectively EUF-ACMA secure relative to $H$.
\qed
\end{theorem}

In order to prove Theorem~\ref{FDH-computable},
we need Lemmas~\ref{lemma1} and \ref{lemma2} in the previous section, and Lemma~\ref{exercise} below.
The last one is Exercise 1.9.21 of Nies's textbook \cite{N09} of algorithmic randomness.%
\footnote{Lemma~\ref{exercise} can be used to prove the non-existence of universal Schnorr test for the notion of Schnorr randomness for an infinite binary sequence.
See Fact 3.5.9 of \cite{N09} for the detail.
}
In Section~\ref{LM} we will prove a modification of Lemma~\ref{exercise}, i.e., Theorem~\ref{exercise-modified}.
The proof of Lemma~\ref{exercise} can be obtained by simplifying the proof of Theorem~\ref{exercise-modified}.

\begin{lemma}\label{exercise}
Let $S$ be an r.e.~subset of $\X$.
Suppose that $\Lm{\osg{S}}<1$ and $\Lm{\osg{S}}$ is a computable real.
Then there exists $\alpha\in\XI$ such that $\alpha$ is computable and $\alpha\notin\osg{S}$.
\qed
\end{lemma}

\begin{proof}[Proof of Theorem~\ref{FDH-computable}]
Suppose that
a signature scheme $\mathrm{\Pi}=(\mathsf{Gen},\mathsf{Sign},\mathsf{Vrfy})$ relative to $\ell$-functions is effectively EUF-ACMA secure in the random oracle model.
Then there exists a computable function $f\colon\N^+\times\N^+\to\N^+$ such that, for all $i,d,n\in\N^+$, if $n\ge f(i,d)$ then
\begin{align*}
  \frac{1}{\#\mathsf{Func}_{\le q_i(n)}^{\ell(n)}}\sum_{G\in\mathsf{Func}_{\le q_i(n)}^{\ell(n)}}\Prob[\mathsf{Sig\text{-}forge}_{\mathcal{A}_i,\mathrm{\Pi}}(n,G)=1]
  \le\frac{1}{n^d},
\end{align*}
where $q_i(n)$ is the maximum value among the running time of $\mathsf{Sign}$, the running time of $\mathsf{Vrfy}$,
and the running time of $\mathcal{A}_i$ on the parameter $n$.
Note that the value $q_i(n)$ can be computed, given $i$ and $n$.
It follows from Lemma~\ref{lemma1} that, for all $i,d,n\in\N^+$, if $n\ge f(i,2d)$ then
\begin{equation}\label{keyineqnn}
\begin{split}
  \#\hspace{-0.5mm}\left\{\hspace{-0.5mm}G\in\mathsf{Func}_{\le q_i(n)}^{\ell(n)}\hspace{-1.0mm}\biggm|
  \hspace{-0.5mm}\Prob[\mathsf{Sig\text{-}forge}_{\mathcal{A}_i,\mathrm{\Pi}}(n,G)=1]>\frac{1}{n^d}\hspace{-0.3mm}\right\}
  <\frac{\#\mathsf{Func}_{\le q_i(n)}^{\ell(n)}}{n^d}.
\end{split}
\end{equation}

For each $i,d,n\in\N^+$ we define a subset $C_{i,d,n}$ of $\X$ as  $C_{\mathcal{A}_i,d,n}$ (see Definition~\ref{def-stest_acma}).
Since $\#\mathsf{Func}_{\le q_i(n)}^{\ell(n)}=2^{\ell(n)\#\{0,1\}^{\le q_i(n)}}$,
it follows from Definition~\ref{def-stest_acma}, (i) of Proposition~\ref{Lebesgue-outer}, and \eqref{keyineqnn} that,
for each $i,d,n\in\N^+$, if $n\ge f(i,2d)$ then
\begin{equation}\label{keyineqnidn}
\begin{split}
  \Lm{\osg{C_{i,d,n}}}=\sum_{s\in C_{i,d,n}}2^{-\abs{s}}
  <\frac{\#\mathsf{Func}_{\le q_i(n)}^{\ell(n)}}{n^d}2^{-\ell(n)\#\{0,1\}^{\le q_i(n)}}=\frac{1}{n^d}.
\end{split}
\end{equation}

We choose a particular computable bijection
\begin{equation*}
  \varphi\colon\N^+\to\{\,(i,d)\mid i\in\N^+\;\&\;d\ge 2\,\},
\end{equation*}
and define $(\varphi_1(m),\varphi_2(m))=\varphi(m)$.
We then define a computable function $g\colon\N^+\to\N^+$ by
$g(m)=\{f(\varphi_1(m),2\varphi_2(m))+1\}^{m+1}$.
For each $m\in\N^+$, we define a subset $C_{m}$ of $\X$ by
\begin{equation}\label{C_m}
  C_m=\bigcup_{n=g(m)}^\infty C_{\varphi_1(m),\varphi_2(m),n}.
\end{equation}
It follows from (iii) of Proposition~\ref{Lebesgue-outer}, \eqref{keyineqnidn}, and Lemma~\ref{lemma2} that, for each $m\in\N^+$,
\begin{equation}\label{LC_mle2m}
\begin{split}
  \Lm{\osg{C_m}}\le\sum_{n=g(m)}^{\infty}\Lm{\osg{C_{\varphi_1(m),\varphi_2(m),n}}}
  <\sum_{n=g(m)}^{\infty}\frac{1}{n^{\varphi_2(m)}}\le\frac{2}{g(m)}\le\frac{1}{2^{m}}.
\end{split}
\end{equation}
We then define $C$
by
\begin{equation}\label{C}
  C=\bigcup_{m=1}^{\infty}C_m.
\end{equation}
Therefore, using (iii) of Proposition~\ref{Lebesgue-outer},
\begin{equation}
  \Lm{\osg{C}}\le\sum_{m=1}^{\infty}\Lm{\osg{C_m}}<\sum_{m=1}^{\infty}\frac{1}{2^{m}}=1.
\end{equation}

Next we show that $C$ is an r.e.~subset of $\X$.
It follows from Definition~\ref{def-stest_acma} that,
given $i$, $d$, and $n$, one can decide the finite subset $C_{i,d,n}$ of $\X$,
since the dyadic rational $\Prob[\mathsf{Sig\text{-}forge}_{\mathcal{A}_i,\mathrm{\Pi}}(n,G)=1]$ is computable,
given $i$, $n$, and
$G\in\mathsf{Func}_{\le q_i(n)}^{\ell(n)}$.
Thus, since $\varphi$ and $g$ are computable functions, it follows from \eqref{C_m} and \eqref{C} that $C$ is an r.e.~subset of $\X$.

We then show that $\Lm{\osg{C}}$ is a computable real.
For each $k\in\N$, we define a finite subset $D_k$ of
$C$
by
\begin{equation*}
  D_k=\bigcup_{m=1}^{k}\bigcup_{n=g(m)}^{g(m)2^{k}-1} C_{\varphi_1(m),\varphi_2(m),n}.
\end{equation*}
Given $k\in\N$, one can decides the finite set $D_k$,
since $\varphi$ and $g$ are computable functions and moreover one can decide the finite set $C_{i,d,n}$, given $i$, $d$, and $n$.
Therefore, given $k\in\N$, one can calculate the dyadic rational $\Lm{\osg{D_k}}$ based on (i) of Proposition~\ref{Lebesgue-outer}.
On the other hand, note that
\begin{equation*}
  C\setminus D_k\subset\left(\bigcup_{m=1}^{k}\bigcup_{n=g(m)2^{k}}^{\infty} C_{\varphi_1(m),\varphi_2(m),n}\right)\cup\bigcup_{m=k+1}^{\infty}C_m.
\end{equation*}
Thus, using (ii) and (iii) of Proposition~\ref{Lebesgue-outer}, \eqref{keyineqnidn}, Lemma~\ref{lemma2}, and \eqref{LC_mle2m} we see that, for each $k\in\N$,
\begin{align*}
  \Lm{\osg{C\setminus D_k}}
  &\le\sum_{m=1}^{k}\sum_{n=g(m)2^{k}}^{\infty}\Lm{\osg{C_{\varphi_1(m),\varphi_2(m),n}}}+\sum_{m=k+1}^{\infty}\Lm{\osg{C_m}}\\
  &<\sum_{m=1}^{k}\frac{2}{g(m)2^k}+\sum_{m=k+1}^{\infty}\frac{1}{2^m}
  \le\sum_{m=1}^{k}\frac{1}{2^{m+k}}+\frac{1}{2^{k}}
  <\frac{1}{2^{k-1}}.
\end{align*}
Therefore, since $\osg{C}=\osg{D_{k+1}}\cup\osg{C\setminus D_{k+1}}$, using (ii) and (iii) of Proposition~\ref{Lebesgue-outer} we have
\begin{equation*}
  \abs{\Lm{\osg{C}}-\Lm{\osg{D_{k+1}}}}\le\Lm{\osg{C\setminus D_{k+1}}}\le 2^{-k}
\end{equation*}
for each $k\in\N$.
Hence, $\Lm{\osg{C}}$ is a computable real.

Now, it follows from Lemma~\ref{exercise} that there exists $H\in\XI$ such that $H$ is computable and $H\notin\osg{C}$.
Since $H$ is computable as an infinite binary sequence,
it is easy to see that
$H$ is also computable as an $\ell$-function.
On the other hand, let $i,d,n\in\N^+$ with $n\ge g(\varphi^{-1}(i,d+1))$.
We then define $m=\varphi^{-1}(i,d+1)$, i.e., $\varphi(m)=(i,d+1)$.
Since $H\notin\osg{C}$ and $n\ge g(m)$, it follows from \eqref{C} and \eqref{C_m} that $H\notin\osg{C_{\varphi_1(m),\varphi_2(m),n}}=\osg{C_{i,d+1,n}}$.
Therefore, based on the identification \eqref{identifylf} of an $\ell$-function with an infinite binary sequence,
we see that the function $H_n\colon\X\to\{0,1\}^{\ell(n)}$ satisfies that
$\Prob[\mathsf{Sig\text{-}forge}_{\mathcal{A}_i,\mathrm{\Pi}}(n,H_n)=1]\le 1/n^{d+1}<1/n^{d}$.
Thus, since the mapping $\N^+\times\N^+\ni (i,d)\mapsto g(\varphi^{-1}(i,d+1))$ is a computable function,
it follows from Definition~\ref{eff-EUACMA} that $\mathrm{\Pi}$ is effectively EUF-ACMA secure relative to $H$.
\end{proof}

\section{Computable analysis on cryptography}
\label{discussion}

In this section, we
show
that the effective security notions introduced in the previous section are a natural alternative to the conventional security notions in modern cryptography.

In Definitions~\ref{EUF-ACMA-secure-rel} and~\ref{EUF-ACMA-secure-ro} for the conventional EUF-ACMA security,
the number $N$ is only required to exist, depending on the adversary $\mathcal{A}$ and the number $d$, that is,
the success probability of the attack by an adversary $\mathcal{A}$ on a security parameter $n$ is required to be less than $1/n^d$
for all sufficiently large $n$, where the lower bound of such $n$ is not required to be computable from $\mathcal{A}$ and $d$.
On the other hand, in Definitions~\ref{eff-EUACMA} and~\ref{eff-EUACMA-ro} for the effective EUF-ACMA security,
it is required that the lower bound $N$ of such $n$ can be computed from the code of $\mathcal{A}$ and $d$.

In modern cryptography based on computational security,
it is important to choose the security parameter $n$ of a cryptographic scheme as small as possible to the extent that the \emph{security requirements} are satisfied,
in order to make the efficiency of the scheme as high as possible.
For that purpose,
it is desirable to be able to calculate a concrete value of $N$, given the code of $\mathcal{A}$ and $d$,
since $N$ gives a lower bound of the security parameter for which the security requirements specified by $\mathcal{A}$ and $d$ are satisfied.
This results in the notion of effective security.

Does the replacement of the conventional security notions by the corresponding effective security notions
bring difficulties to modern cryptography over all ?
We do not think so.
It would seem plausible that all the conventional security notions can be replaced by the corresponding effective security notions in modern cryptography with little cost.
As an example,
let us consider the EUF-ACMA security of the RSA-FDH signature scheme under the RSA assumption and its effective counterpart.
Let $\mathsf{Succ}_{\mathcal{A}}^{\mathsf{RSA}}(n)$ be
the success probability of an algorithm $\mathcal{A}$ in solving the RSA problem on a security parameter $n$.
On the one hand, the (conventional) RSA assumption is defined as the condition that, for all probabilistic polynomial-time algorithms $\mathcal{A}$ and all $d\in\N^+$
there exists $N\in\N^+$ such that, for all $n\ge N$,
\begin{equation*}
  \mathsf{Succ}_{\mathcal{A}}^{\mathsf{RSA}}(n)\le \frac{1}{n^d}.
\end{equation*}
On the other hand, the \emph{effective} RSA assumption is defined as the condition that
there exists a computable function $f\colon\N^+\times\N^+\to\N^+$ such that, for all $i,d,n\in\N^+$,
if $n\ge f(i,d)$ then
\begin{equation*}
  \mathsf{Succ}_{\mathcal{A}_i}^{\mathsf{RSA}}(n)\le \frac{1}{n^d},
\end{equation*}
where $\mathcal{A}_i$ is the $i$th algorithm in a particular recursive enumeration of all probabilistic polynomial-time algorithms.
Now, recall the following theorem.

\begin{theorem}[Bellare and Rogaway~\cite{BR93}]\label{RSA-FDH}
RSA-FDH is EUF-ACMA secure in the random oracle model under the RSA assumption.
\qed
\end{theorem}

By analyzing the proof of Theorem~\ref{RSA-FDH} given in \cite{BR93}, we can see that the following effective version of Theorem~\ref{RSA-FDH} holds.
We can do this task very easily, compared with
the non-triviality
of the original proof itself.

\begin{theorem}\label{effective_RSA-FDH}
RSA-FDH is \emph{effectively} EUF-ACMA secure in the random oracle model under the \emph{effective} RSA assumption.
\qed
\end{theorem}

Note that the effective RSA assumption seems more difficult to prove than the RSA assumption.
However, in modern cryptography based on computational security,
we must make a computational assumption, such as the RSA assumption, somehow to guarantee the security of a cryptographic scheme.
Since making any computational assumption does not cost at all in the development of theory of cryptography,
making the effective RSA assumption instead of the RSA assumption would not seem to bring any trouble to modern cryptography.
In this manner, we would expect that
all the conventional security notions can be replaced by the corresponding effective security notions in modern cryptography with little cost.
Thus, it would seem plausible that we can easily reconstruct the theory of cryptography based on the effective security notions instead of the conventional security notions.

In the above,
we consider the validity of the effective security notions in modern cryptography.
Actually,
it would seem more natural to require that
the functions $f\colon\N^+\times\N^+\to\N^+$ in Definitions~\ref{eff-EUACMA} and~\ref{eff-EUACMA-ro} are polynomial-time computable rather than simply computable.
We call this type of effective security \emph{polynomial-time effective security}.
Conjecture~\ref{conjecture-future} below is a polynomial-time effective version of Conjecture~\ref{FDH-conjecture1},
and states that \emph{the security in the random oracle model implies one in the standard model}.
In the future, it would be challenging to prove Conjecture~\ref{conjecture-future} (or its appropriate modification)
with identifying an appropriate computational assumption $\mathsf{COMP}$ and an appropriate nontrivial condition $\mathcal{C}$ on a signature scheme $\mathrm{\Pi}$.

\begin{conjecture}\label{conjecture-future}
Let $\ell(n)$ be a polynomial.
Suppose that a signature scheme $\mathrm{\Pi}$ relative to $\ell$-functions is \emph{polynomial-time effectively} EUF-ACMA secure in the random oracle model.
Under the assumption $\mathsf{COMP}$, if $\mathrm{\Pi}$ satisfies the condition $\mathcal{C}$,
then there exists a \emph{polynomial-time computable} $\ell$-function (or a \emph{polynomial-time computable} family of $\ell$-functions)
relative to which $\mathrm{\Pi}$ is \emph{polynomial-time effectively} EUF-ACMA secure.
\qed
\end{conjecture}

Note that the computational assumption $\mathsf{COMP}$ should be needed in Conjecture~\ref{conjecture-future}.
Without this assumption, Conjecture~\ref{conjecture-future} implies that the complexity class $P$ is a proper subclass of the class $NP$,
unless no signature scheme $\mathrm{\Pi}$ satisfies the condition $\mathcal{C}$.
Thus, Conjecture~\ref{conjecture-future} without the computational assumption $\mathsf{COMP}$ would become very difficult to prove.
Actually, in the random oracle methodology, the random oracle is instantiated by a concrete cryptographic hash function such as the SHA hash functions
(without adequate theoretical reason).
Thus, from a theoretical point of view, it would seem reasonable to assume at least the existence of a \emph{collision resistant hash function}%
\footnote{See \cite[Chapter 4]{KL07} for the detail of collision resistant hash function.}
as the computational assumption $\mathsf{COMP}$.

Computable analysis \cite{PR89,W00} is a branch of computation theory
which studies the computability and the computational complexity of
mathematical notions appearing in analysis.
It is closely related to algorithmic randomness.
In particular, computable analysis considers the notion of the \emph{effective convergence} of a sequence of reals,
where a sequence $\{a_n\}_{n\in\N^+}$ of reals is called \emph{converges effectively} to a real $\alpha$
if there exists a computable function $f\colon\N^+\to\N^+$ such that, for every $N,n\in\N^+$, if $n\ge f(N)$ then $\abs{a_n-\alpha}<1/N$.
On the one hand,
we can see that the form of the definition of the convergence of a sequence of reals in analysis well corresponds to the definitions of security
in modern cryptography based on computational security,
such as Definitions~\ref{EUF-ACMA-secure-rel} and~\ref{EUF-ACMA-secure-ro}.
On the other hand,
we can see that the notion of the effective convergence of a sequence of reals in computable analysis well corresponds to
the effective security notions introduced in Definitions~\ref{eff-EUACMA} and~\ref{eff-EUACMA-ro}.
Thus, the replacement of the conventional security notion by the corresponding effective security notion moder cryptography is just regarded
as performing computable analysis over cryptography.
The results of the previous section shows that
\emph{performing computable analysis over cryptography results in the secure instantiation of the random oracle}.

In what follows we continue to perform computable analysis over cryptography
by introducing the notion of \emph{effective hardness} for computational problems,
whose hardness is used as a computational assumption to prove the security of a cryptographic scheme in moder cryptography.
In particular, we consider the \emph{discrete logarithm problem} and the \emph{Diffie-Hellman problem} in the \emph{generic group model},
and investigate the secure instantiation of the generic group, i.e., a random encoding of the group elements,
in what follows.

\section{The discrete logarithm problem in the generic group model}
\label{DL}

In this section we
review
the discrete logarithm problem in the generic group model.
For the discrete logarithm problem \emph{in the standard model} and its related problems, such as the
Diffie-Hellman
problem in the standard model,
we refer the reader to Katz and Lindell \cite[Chapter 7]{KL07}.
Shoup~\cite{S97} introduced the notion of \emph{generic algorithm} to study the computational complexity of the
discrete logarithm and related problems in the \emph{generic group model}, where the generic algorithm does not exploit any special properties
of the encodings of group elements, other than the property that each group element is encoded as a unique binary string.
Formally, a generic algorithm is defined as follows.

For any integer $N\ge 2$, we denote by $\Z_N$ the \emph{additive group of integers modulo $N$} and sometimes the set $\{0,1,\dots,N-1\}$.
For any $n\in\N^+$, an \emph{encoding function into $n$ bitstrings} is a bijective function
mapping $\{0,1,\dots,2^n-1\}$ to $\{0,1\}^n$.
Let $N$ be a positive integer with $N\le 2^n$,
and let $\mathcal{G}_N$ be the set of all finite cyclic groups $G$ of order $N$ with $G\subset\{0,1\}^n$.
Given an arbitrary finite cyclic group $G$ of order $N$,
we can represent each element of $G$ by a unique $n$ bits string, since $N\le 2^n$.
Thus, every finite cyclic group $G$ of order $N$ is in $\mathcal{G}_N$ in $\binom{2^n}{N}$ distinct representations.
Recall that every finite cyclic group $G$ of order $N$ can be  isomorphic to the additive group $\Z_N$ based on a generator of $G$.
Thus, we see that, for every pair of a finite cyclic group $G\in\mathcal{G}_N$ and its generator $g$,
there is an encoding function $\sigma$ into $n$ bitstrings such that $\Z_N$ is isomorphic to $G$ via $\sigma$ and
$\sigma(1)=g$.
Conversely, for every encoding function $\sigma$ into $n$ bitstrings,
by defining a binary operation $\circ\colon\sigma(\Z_N)\times\sigma(\Z_N)\to\sigma(\Z_N)$ by
\begin{equation*}
  \sigma(x)\circ\sigma(y):=\sigma(x+y),
\end{equation*}
the set $\sigma(\Z_N)$ becomes a finite cyclic group in $\mathcal{G}_N$ with the generator $\sigma(1)$ and
$\Z_N$ is isomorphic to $\sigma(\Z_N)$ via $\sigma$.
In this manner,
there is a surjective mapping from an encoding function $\sigma$ into $n$ bitstrings to a pair of a finite cyclic group $G\in\mathcal{G}_N$ and its generator.
If we restrict the domain of definition of encoding functions into $n$ bitstrings to $\Z_N$, the mapping becomes bijective.

A \emph{generic algorithm} is a probabilistic oracle Turing machine $\mathcal{A}$ which behaves as follows \cite{S97,MW98}:
Let $n\in\N^+$, and let $\sigma$ be an encoding function into $n$ bitstrings and $N$ a positive integer with $N\le 2^n$.
\renewcommand{\labelenumi}{(\roman{enumi})}
\begin{enumerate}
\item $\mathcal{A}$ takes as input a list $\sigma(x_1),\dots,\sigma(x_k)$ with $x_1,\dotsc,x_k\in\Z_{N}$,
  as well as (the binary representations of) $N$ and its prime factorization.
\item As $\mathcal{A}$ is executed, it is allowed to make calls to oracles
  which compute the functions $add\colon\sigma(\Z_N)\times\sigma(\Z_N)\to\sigma(\Z_N)$ and $inv\colon\sigma(\Z_N)\to\sigma(\Z_N)$ with
  \begin{equation*}
    add(\sigma(x),\sigma(y))=\sigma(x+y)\;\text{ and }\;inv(\sigma(x))=\sigma(-x).
  \end{equation*}
\item Eventually, $\mathcal{A}$ halts and outputs a finite binary string, denoted by
  \begin{equation*}
    \mathcal{A}(N;\sigma(x_1),\dots,\sigma(x_k)).
  \end{equation*}
\end{enumerate}

Consider the following experiment for a polynomial-time generic algorithm $\mathcal{A}$,
a parameter $n$,
and a positive integer $N\le 2^n$:

\begin{quote}
\textbf{The discrete logarithm experiment $\mathsf{DLog}_{\mathcal{A}}(n,N)$:}
\vspace*{-1.5mm}
\renewcommand{\labelenumi}{\arabic{enumi}.}
\emph{
\begin{enumerate}
\item Generate an encoding function $\sigma$ into $n$ bitstrings uniformly.
\item Generate $x\in\Z_{N}$ uniformly.
\item The output of the experiment is defined to be $1$ if $\mathcal{A}(N;\sigma(1),\sigma(x))=x$ and $0$ otherwise.
\end{enumerate}
}
\end{quote}

Note here that $x\in\Z_{N}$ is the discrete logarithm of $\sigma(x)$ with respect to the generator $\sigma(1)$ in the finite cyclic group $\sigma(\Z_N)$ of order $N$.
Thus in the experiment, given a generator $\sigma(1)$ of a finite cyclic group $\sigma(\Z_N)$ and an element $\sigma(x)$ of $\sigma(\Z_N)$,
the generic algorithm $\mathcal{A}$ tries to calculate the discrete logarithm $x$ of $\sigma(x)$
while making calls to oracles which compute the functions $add$ and $inv$.
Shoup~\cite{S97} showed the following lower bound for the complexity of the discrete logarithm problem in the generic group model.

\begin{theorem}[Shoup~\cite{S97}]\label{Shoup}
There exists $C\in\N^+$ such that, for every generic algorithm $\mathcal{A}$, $n\in\N^+$, and $N$ with $2\le N\le 2^n-1$,
\begin{equation*}
  \Prob[\mathsf{DLog}_{\mathcal{A}}(n,N)=1]\le\frac{C m^2}{p},
\end{equation*}
where $p$ is the largest prime divisor of $N$
and $m$ is the maximum number of the oracle queries among all the computation paths of $\mathcal{A}$.  
\qed
\end{theorem}

Theorem~\ref{Shoup} says that any generic algorithm that solves with nonzero constant probability the discrete logarithm problem
in finite cyclic groups of order $N$
must perform at least $\Omega(\sqrt{p})$ group operations (i.e., oracle queries).

In what follows, we show that the generic group, i.e, the random encoding function $\sigma$ into $n$ bitstrings,
used in the discrete logarithm problem
can be instantiated by a deterministic and computable one while keeping the computational hardness originally proved in the generic group model,
as in Theorem~\ref{Shoup}.
Before that, we develop the Lebesgue outer measure on families of encoding functions in the next section.

\section{Lebesgue outer measure on families of encoding functions}
\label{LM}

For each $n\in\N^+$, we denote by $\mathsf{Encf}_{n}$ the set of all encoding functions into $n$ bitstrings.
Note that $\#\mathsf{Encf}_{n}=(2^n)!$.
A \emph{family of encoding functions} is an infinite  sequence $\{\sigma_n\}_{n\in\N^+}$ such that $\sigma_n$ is an encoding function into $n$ bitstrings for all $n\in\N^+$.
A family of encoding functions serves as an instantiation of an infinite sequence of the generic groups, i.e., random encoding functions, over all security parameters.
We denote by $\mathsf{Encf}^{\infty}$ the set of all families of encoding functions.
Namely,
\begin{equation*}
  \mathsf{Encf}^{\infty}
  :=\prod_{k=1}^{\infty}\mathsf{Encf}_{k}=\mathsf{Encf}_{1}\times\mathsf{Encf}_{2}\times\mathsf{Encf}_{3}\times\dotsm\dotsm.
\end{equation*}

On the other hand,
a \emph{finite family of encoding functions} is a finite sequence $s=(\sigma_1,\dots,\sigma_n)$
such that $\sigma_k$ is an encoding function into $k$ bitstrings for
all
$k=1,\dots,n$.
Here, $n$ is called the \emph{length} of $s$ and denoted by $\abs{s}$.
A finite family of encoding functions is an initial segment (a prefix) of a family of encoding functions.
For each $n\in\N$, we denote by $\mathsf{Encf}^n$ the set of all finite families of encoding functions of length $n$.
Namely,
\begin{equation*}
  \mathsf{Encf}^n:=\prod_{k=1}^n\mathsf{Encf}_{k}=\mathsf{Encf}_{1}\times\dots\times\mathsf{Encf}_{n}.
\end{equation*}
Note that $\mathsf{Encf}^0=\{\lambda\}$ where $\lambda:=()$ is the \emph{empty sequence}.
We denote by $\mathsf{Encf}^*$ the set of all finite families of encoding functions, i.e., $\mathsf{Encf}^*:=\bigcup_{n=0}^{\infty}\mathsf{Encf}^n$.
For any sequences $s=(\sigma_1,\dots,\sigma_n)$ and $t=(\tau_1,\dots,\tau_m)$ in $\mathsf{Encf}^*$,
we say that $s$ is a \emph{prefix} of $t$ if
$n\le m$ and $\sigma_k=\tau_k$ for all $k\le n$.
A subset $P$ of $\mathsf{Encf}^*$ is called \emph{prefix-free} if no sequence in $P$ is a prefix of another sequence in $P$.

In what follows we use the notion of Lebesgue outer measure on $\mathsf{Encf}^{\infty}$, which is defined as follows.
For any sequence $s=(\sigma_1,\dots,\sigma_n)\in\mathsf{Encf}^*$,
$I(s)$ is defined as the set of all families $\{\tau_k\}_{k\in\N^+}$ of encoding functions for which $\sigma_k=\tau_k$ for all $k\le n$,
and $\abs{I(s)}$ is defined by
\begin{equation*}
  \abs{I(s)}:=\prod_{k=1}^{n}\frac{1}{\#\mathsf{Encf}_{k}}=\frac{1}{\#\mathsf{Encf}_{1}\times\dots\times\#\mathsf{Encf}_{n}}.
\end{equation*}
Note that $I(\lambda)=\mathsf{Encf}^\infty$ and $\abs{I(\lambda)}=1$.
\emph{Lebesgue outer measure $\mathcal{L}$ on $\mathsf{Encf}^{\infty}$} is a function mapping any subset $A$ of $\mathsf{Encf}^{\infty}$ to a non-negative real,
and is defined by
\begin{equation*}
  \Lm{A}:=\inf\sum_{n=1}^\infty \abs{I(s_n)},
\end{equation*}
where the infimum extends over all infinite sequences $s_1,s_2,\dotsc\in\mathsf{Encf}^*$ for which $A\subset\bigcup_{n=1}^\infty I(s_n)$. 

In what follows,
we use the properties of $\mathcal{L}$ presented in Proposition~\ref{Lebesgue-outer_ggm} below.
For any subset $T$ of $\mathsf{Encf}^*$, we denote by $\osg{T}$ the set $\bigcup_{s\in T}I(s)$.

\begin{proposition}%
\label{Lebesgue-outer_ggm}\hfill
\begin{enumerate}
\item For every prefix-free set $P\subset\mathsf{Encf}^*$,
\begin{equation*}
  \Lm{\osg{P}}=\sum_{s\in P}\abs{I(s)}.
\end{equation*}
  Therefore
  $\Lm{\emptyset}=\Lm{\osg{\emptyset}}=0$ and
  $\Lm{\mathsf{Encf}^{\infty}}=\Lm{\osg{\{\lambda\}}}=1$.
\item $\Lm{A}\le\Lm{B}$ for every sets $A\subset B\subset\mathsf{Encf}^{\infty}$.
\item $\Lm{\bigcup_{i}A_i}\le \sum_{i}\Lm{A_i}$ for every sequence $\{A_i\}_{i\in\N}$ of subsets of $\mathsf{Encf}^{\infty}$.
\item $\Lm{\bigcup_{i}\osg{P_i}}=\sum_{i}\Lm{\osg{P_i}}$ for every finite or infinite sequence $\{P_i\}_i$ of subsets of $\mathsf{Encf}^*$ such that
   $\osg{P_i}\cap\osg{P_j}=\emptyset$ for every $i\neq j$.\qed
\end{enumerate}
\end{proposition}

For any subset $S$ of $\mathsf{Encf}^*$,
we say that $S$ is \emph{recursively enumerable} (\emph{r.e.}, for short) if
there exists a deterministic Turing machine which on every input $s\in\mathsf{Encf}^*$ halts if and only if $s\in S$.
Note here that any sequence in $\mathsf{Encf}^*$ is a finite object, which can be represented as a finite binary string,
and thus can be manipulated by a Turing machine.
Finally, a family $\{\sigma_n\}_{n\in\N^+}$ of encoding functions is called \emph{computable}
if there exists a deterministic Turing machine which on every input $(n,x)$ with
$x\in\{0,1,\dots,2^n-1\}$
halts and outputs $\sigma_n(x)$.

Theorem~\ref{exercise-modified} below plays a crucial role in what follows.
It is a modification of
Lemma~\ref{exercise}.
We can prove this theorem based on the properties
of $\mathcal{L}$ in Proposition~\ref{Lebesgue-outer_ggm},
as well as the computability of the mapping $\N^+\ni n\mapsto\#\mathsf{Encf}_n$.

\begin{theorem}\label{exercise-modified}
Let $S$ be an r.e.~subset of $\mathsf{Encf}^*$.
Suppose that $\Lm{\osg{S}}<1$ and $\Lm{\osg{S}}$ is a computable real.
Then there exists
a computable family
of encoding functions which is not in $\osg{S}$.
\end{theorem}

\begin{proof}
We define $F\colon\mathsf{Encf}^*\to[0,1]$ by
$F(t)=\Lm{\osg{S}\cap I(t)}$.
First, we show that the real-valued function $F$ is computable, i.e.,
there exists a computable function
$f\colon\mathsf{Encf}^*\times\N\to\Q$ such that
\begin{equation}\label{absF-f<2}
  \abs{F(t)-f(t,k)} < 2^{-k}
\end{equation}
for all $t\in\mathsf{Encf}^*$ and $k\in\N$.

Let $n\in\N$.
Since
$\bigcup_{t\in\mathsf{Encf}^n}I(t)=\mathsf{Encf}^\infty$
we have
\begin{equation*}
  \bigcup_{t\in\mathsf{Encf}^n}\osg{S}\cap I(t)=\osg{S}
\end{equation*}
and
$\left(\osg{S}\cap I(t)\right)\cap\left(\osg{S}\cap I(t')\right)=\emptyset$
for any distinct $t,t'\in\mathsf{Encf}^n$.
Note that, for every $t\in\mathsf{Encf}^*$, there is $S'\subset\mathsf{Encf}^*$ such that $\osg{S}\cap I(t)=\osg{S'}$.%
\footnote{As such $S'$, the set $T\cup\{s\in S\mid\text{$t$ is a prefix of $s$}\}$ suffices,
where $T=\{t\}$ if there is a prefix $s\in S$ of $t$ and $T=\emptyset$ otherwise.}
It follows from (iv) of Proposition~\ref{Lebesgue-outer_ggm} that
\begin{equation}\label{sumFt=LS}
  \sum_{u\in\mathsf{Encf}^n}F(u)=\Lm{\osg{S}}
\end{equation}
for every $n\in\N$.

Since $S$ is an r.e.~set, there is a deterministic Turing machine which enumerates $S$, i.e.,
there is a deterministic Turing machine which on every input $m\in\N^+$ outputs a finite subset $S_m$ of $S$,
where $S_m\subset S_{m+1}$ for every $m\in\N^+$ and $\bigcup_{m=1}^\infty S_m=S$.
Therefore, for each $t\in\mathsf{Encf}^*$,
we have $\osg{S_m}\cap I(t)\subset \osg{S_{m+1}}\cap I(t)$ for every $m\in\N^+$ and $\bigcup_{m=1}^\infty \left(\osg{S_m}\cap I(t)\right)=\osg{S}\cap I(t)$.
Using (ii) and (iv) of Proposition~\ref{Lebesgue-outer_ggm} it is easy to show that, for each $t\in\mathsf{Encf}^*$,
$\Lm{\osg{S_m}\cap I(t)}\le\Lm{\osg{S_{m+1}}\cap I(t)}$
for every $m\in\N^+$ and $\lim_{m\to\infty}\Lm{\osg{S_m}\cap I(t)}=F(t)$.
It follows from \eqref{sumFt=LS} that, for each $n\in\N$,
\begin{equation}\label{sumleLS}
\begin{split}
  \sum_{u\in\mathsf{Encf}^{n}}\Lm{\osg{S_m}\cap I(u)}
  \le F(t)+\sum_{u\in\mathsf{Encf}^{n}\text{ and }u\neq t}\Lm{\osg{S_m}\cap I(u)}
  \le \Lm{\osg{S}}
\end{split}
\end{equation}
for every $m\in\N^+$ and $t\in\mathsf{Encf}^n$, and
\begin{equation}\label{limsumE=LS}
  \lim_{m\to\infty}\sum_{u\in\mathsf{Encf}^{n}}\Lm{\osg{S_m}\cap I(u)}=\Lm{\osg{S}}.
\end{equation}

Note that, any given finite set $P\subset\mathsf{Encf}^*$, one can compute a finite prefix-free set $Q\subset\mathsf{Encf}^*$ such that $\osg{Q}=\osg{P}$.
It follows from (i) of Proposition~\ref{Lebesgue-outer_ggm} and the computability of the mapping $\N^+\ni l\mapsto\#\mathsf{Encf}_l$ that,
any given $t\in\mathsf{Encf}^*$ and $m\in\N^+$, one can compute the rational $\Lm{\osg{S_m}\cap I(t)}$.
Therefore, any given $n\in\N$ and $m\in\N^+$, one can compute the rational $\sum_{u\in\mathsf{Encf}^{n}}\Lm{\osg{S_m}\cap I(u)}$.

Now, since $\Lm{\osg{S}}$ is a computable real by the assumption, there exists a computable function $g\colon\N\to\Q$ such that
\begin{equation}\label{absLS-gk<2k}
  \abs{\Lm{\osg{S}}-g(k)} < 2^{-k}
\end{equation}
for all $k\in\N$.
It follows from \eqref{limsumE=LS} that there exists a computable function $h\colon\mathsf{Encf}^*\times\N\to\N^+$ such that,
for every $t\in\mathsf{Encf}^*$ and $k\in\N$,
\begin{equation*}
  g(k)-2^{-k}<\sum_{u\in\mathsf{Encf}^{\abs{t}}}\Lm{\osg{S_{h(t,k)}}\cap I(u)}.
\end{equation*}
But, by \eqref{sumleLS} and \eqref{absLS-gk<2k}, the right-hand side is at most
\begin{equation*}
  F(t)+\sum_{u\in\mathsf{Encf}^{\abs{t}}\text{ and }u\neq t}\Lm{\osg{S_{h(t,k)}}\cap I(u)} < g(k)+2^{-k}.
\end{equation*}
Thus we define a function $f\colon\mathsf{Encf}^*\times\N\to\Q$ by
\begin{equation*}
  f(t,k)=g(k)-\sum_{u\in\mathsf{Encf}^{\abs{t}}\text{ and }u\neq t}\Lm{\osg{S_{h(t,k)}}\cap I(u)}.
\end{equation*}
We then see that the rational-valued function $f$ is computable and \eqref{absF-f<2} holds, as desired.

Next, we construct a computable family $\{\sigma_n\}_{n\in\N^+}$ of encoding functions such that
\begin{equation}\label{prop_n}
  F((\sigma_1,\dots,\sigma_m))
  <\prod_{k=1}^{m}\frac{1}{\#\mathsf{Encf}_{k}}
\end{equation}
holds for all $m\in\N$.
We do this
by the recursive procedure given below.
First, since
\begin{equation*}
  \bigcup_{\tau\in\mathsf{Encf}_{m+1}}I((\tau_1,\dots,\tau_{m},\tau))=I((\tau_1,\dots,\tau_m))
\end{equation*}
holds for every $m\in\N$ and $(\tau_1,\dots,\tau_m)\in\mathsf{Encf}^m$, we note by (iv) of Proposition~\ref{Lebesgue-outer_ggm} that
\begin{equation}\label{LSIsumm+1=m}
\begin{split}
  \sum_{\tau\in\mathsf{Encf}_{m+1}}F((\tau_1,\dots,\tau_{m},\tau))
  =F((\tau_1,\dots,\tau_m))
\end{split}
\end{equation}
for every $m\in\N$ and $(\tau_1,\dots,\tau_m)\in\mathsf{Encf}^m$.
Let $s_n=(\sigma_1,\dots,\sigma_n)\in\mathsf{Encf}^n$ for each $n\in\N$.
Then the recursive procedure is given as follows.

Initially, we set $n:=0$ and $s_n:=\lambda$.
Then, obviously, the property \eqref{prop_n} holds for $m=n$, which is precisely the assumption $\Lm{\osg{S}}<1$ of the theorem.

For an arbitrary $n\in\N$, assume that we have constructed $s_n=(\sigma_1,\dots,\sigma_n)$ and \eqref{prop_n} holds for $m=n$.
It follows from \eqref{LSIsumm+1=m} with $m=n$
that
\begin{equation}\label{LSI}
  F((\sigma_1,\dots,\sigma_n,\tau_0))
  <\prod_{k=1}^{n+1}\frac{1}{\#\mathsf{Encf}_{k}}
\end{equation}
for some $\tau_0\in\mathsf{Encf}_{n+1}$.
Since $F$ is a computable real function,
by computing the approximation of $F((\sigma_1,\dots,\sigma_n,\tau))$ with an arbitrary precision for each $\tau\in\mathsf{Encf}_{n+1}$,
one can find $\tau_0$ for which \eqref{LSI} holds, and then set $\sigma_{n+1}:=\tau_0$ and $s_{n+1}:=(\sigma_1,\dots,\sigma_n,\tau_0)$.  
It follows that \eqref{prop_n} holds for $m=n+1$.

Thus, any given $n\in\N^+$, one can compute $\sigma_n$ by the above procedure.
This implies that the family $\{\sigma_n\}_{n\in\N^+}$ of encoding functions is computable.

Now, assume contrarily that $\{\sigma_n\}_{n\in\N^+}\in\osg{S}$.
Then there is $n\in\N$ such that $(\sigma_1,\dots,\sigma_n)\in S$.
It follows that
\begin{align*}
  F((\sigma_1,\dots,\sigma_n))
  =\Lm{I((\sigma_1,\dots,\sigma_n))}
  =\prod_{k=1}^{n}\frac{1}{\#\mathsf{Encf}_{k}}.
\end{align*}
However, this contradicts \eqref{prop_n} with $m=n$.
Hence we have $\{\sigma_n\}_{n\in\N^+}\notin\osg{S}$, and the proof is completed.
\end{proof}

\section{Effective hardness and secure instantiation of the generic group}
\label{DL2}

In this section we introduce the notion of \emph{effective hardness} for the discrete logarithm problem,
and then show that the generic group used in the problem can be instantiated by a deterministic and computable one
while keeping the computational hardness.
For that purpose, we first translate Theorem~\ref{Shoup} into the form well used as a computational hardness assumption for a cryptographic scheme in cryptography.

Consider the following experiment for a polynomial-time generic algorithm $\mathcal{A}$, a parameter $n$, and an encoding function $\sigma$ into $n$ bitstrings:

\begin{quote}
\textbf{The discrete logarithm experiment $\mathsf{DLog}_{\mathcal{A}}(n,\sigma)$:}
\vspace*{-1.5mm}
\renewcommand{\labelenumi}{\arabic{enumi}.}
\emph{
\begin{enumerate}
\item Generate an $n$-bit prime $p$ uniformly.
\item Generate $x\in\Z_{p}$ uniformly.
\item The output of the experiment is defined to be $1$ if $\mathcal{A}(p;\sigma(1),\sigma(x))=x$ and $0$ otherwise.
\end{enumerate}
}
\end{quote}

In the experiment above, we consider the discrete logarithm problem in the finite cyclic group $\sigma(\Z_p)$ of a \emph{prime} order $p$.
The reason for choosing a prime order is to minimize the probability of the generic algorithm $\mathcal{A}$ solving the discrete logarithm problem.
This can be checked from the form of Theorem~\ref{Shoup}.

The hardness of the discrete logarithm problem in the generic group model is then formulated as follows.

\begin{definition}\label{def-hard-DLog-ggm}
We say that \emph{the discrete logarithm problem is hard in the generic group model}
if for all polynomial-time generic algorithms $\mathcal{A}$ and all $d\in\N^+$ there exists $N\in\N^+$ such that, for all $n\ge N$,
\begin{equation}\label{randomized}
  \frac{1}{\#\mathsf{Encf}_{n}}\sum_{\sigma\in\mathsf{Encf}_{n}} \Prob[\mathsf{DLog}_{\mathcal{A}}(n,\sigma)=1]\le\frac{1}{n^d}.
\end{equation}
\qed
\end{definition}

Note that, in the left-hand side of \eqref{randomized}, the probability is averaged over all encoding functions into $n$ bitstrings.
This
results in a \emph{random} encoding function into $n$ bitstrings, i.e., the \emph{generic group}.

In this paper we consider a stronger notion of the hardness of the discrete logarithm problem than that given by Definition~\ref{def-hard-DLog-ggm} above.
This stronger notion, called the \emph{effective hardness} of the discrete logarithm problem, is defined as follows:
We first choose a particular recursive enumeration $\mathcal{A}_1,\mathcal{A}_2,\mathcal{A}_3,\dotsc$ of all polynomial-time generic algorithms.
It is easy to show that such an enumeration exists.
In fact,
the $k$th polynomial-time generic algorithm $\mathcal{A}_k$ can be chosen as
a generic algorithm obtained by executing the $k$th generic algorithm $\mathcal{M}_k$ in at most $n^k+k$ steps,
where $n$ is the length of the input of $\mathcal{M}_k$.
We use this specific enumeration as the standard one throughout the rest of this paper.

\begin{definition}\label{def-eff-hard-DLog-ggm}
We say that \emph{the discrete logarithm problem is effectively hard in the generic group model}
if there exists a computable function $f\colon\N^+\times\N^+\to\N^+$ such that, for all $i,d,n\in\N^+$, if $n\ge f(i,d)$ then
\begin{equation*}
  \frac{1}{\#\mathsf{Encf}_{n}}\sum_{\sigma\in\mathsf{Encf}_{n}} \Prob[\mathsf{DLog}_{\mathcal{A}_i}(n,\sigma)=1]\le\frac{1}{n^d}.
\end{equation*}
\qed
\end{definition}

\begin{theorem}\label{eff-hard-DLog-ggm}
The discrete logarithm problem is effectively hard in the generic group model.
\qed
\end{theorem}

In order to prove Theorem~\ref{eff-hard-DLog-ggm}, we need the following lemma.%
\footnote{In order to prove Theorem~\ref{eff-hard-DLog-ggm},
it is suffice to use the inequality $2^n\ge n^d$ which holds for all $n\ge ((d+1)/\ln 2)^{d+1}$ and not for all $n\ge d^2$ as in Lemma~\ref{for-efficient-lower-bound}.
The former follows immediately from the inequality $e^x\ge x$ which holds for all $x\in\R$.
However, we prefer a more ``insightful'' polynomial lower bound $n\ge d^2$ than the super-exponential lower bound $n\ge ((d+1)/\ln 2)^{d+1}$.
See Section~\ref{conclusion-ggm} for
further remarks.}

\begin{lemma}\label{for-efficient-lower-bound}
Let $d\ge 4$.
Then $2^n\ge n^d$ for all $n\ge d^2$.
\end{lemma}

\begin{proof}
We first show that
\begin{equation}\label{2egee2}
  2^d\ge d^2
\end{equation}
by induction.
Obviously, $2^k\ge k^2$ holds for $k=4$.
For an arbitrary $k\ge 4$, assume that $2^k\ge k^2$ holds.
Then $2^{k+1}\ge 2k^2\ge (k+1)^2$,
where the second inequality follows from the inequality $\sqrt{2}x\ge x+1$ for all $x\ge\sqrt{2}+1$.
Thus \eqref{2egee2} holds.

Now, we show that
\begin{equation}\label{2ngend}
  2^n\ge n^d
\end{equation}
holds for all $n\ge d^2$ by induction.
First, it follows from \eqref{2egee2} that $2^{d^2}\ge (d^2)^d$, which implies that \eqref{2ngend} holds for $n=d^2$.
For an arbitrary $k\ge d^2$, assume that \eqref{2ngend} holds for $n=k$.
We note that $2^{1/d}-1\ge \ln 2/d> 1/(2d)> 1/d^2$, where the first inequality follows from the mean-value theorem.
Then, since $k\ge d^2$, we see that $2^{(k+1)/d}\ge 2^{1/d}k\ge k+k/d^2\ge k+1$.
This implies that \eqref{2ngend} holds for $n=k+1$.
Thus, \eqref{2ngend} holds for all $n\ge d^2$.
\end{proof}

\begin{proof}[Proof of Theorem~\ref{eff-hard-DLog-ggm}]
Let $k\in\N^+$, and consider the $k$th generic algorithm $\mathcal{A}_k$.
Since the number of oracle queries along any computation path of $\mathcal{A}_k$ is bounded to the above by $n^k+k$,
it follows from Theorem~\ref{Shoup} that there exists $C\in\N^+$ such that, for every $n\in\N^+$ and $n$-bit prime $p$,
\begin{equation*}
  \Prob[\mathsf{DLog}_{\mathcal{A}_k}(n,p)=1]\le\frac{C (n^k+k)^2}{p}\le\frac{C (n^k+k)^2}{2^{n-1}}.
\end{equation*}
Therefore, for every $n\ge\max\{k,2C\}$,
\begin{equation}\label{lemmma-eff1}
  \frac{1}{\#\mathsf{Encf}_{n}}\sum_{\sigma\in\mathsf{Encf}_{n}} \Prob[\mathsf{DLog}_{\mathcal{A}_k}(n,\sigma)=1]\le\frac{n^{2k+1}}{2^{n}}.
\end{equation}
Note by Lemma~\ref{for-efficient-lower-bound} that, for each $d\in\N^+$,
\begin{equation}\label{lemmma-eff2}
  \frac{n^{2k+1}}{2^n}\le\frac{1}{n^d}
\end{equation}
for every $n\ge(2k+d+1)^2$.

Thus we define a function $f\colon\N^+\times\N^+\to\N^+$ by
\begin{equation*}
  f(k,d)=\max\{(2k+d+1)^2,2C\}.
\end{equation*}
Then $f$ is computable, and it follows from \eqref{lemmma-eff1} and \eqref{lemmma-eff2} that,
for all $k,d,n\in\N^+$, if $n\ge f(k,d)$ then
\begin{equation*}
  \frac{1}{\#\mathsf{Encf}_{n}}\sum_{\sigma\in\mathsf{Encf}_{n}} \Prob[\mathsf{DLog}_{\mathcal{A}_k}(n,\sigma)=1]\le\frac{1}{n^d}.
\end{equation*}
This completes the proof.
\end{proof}

The hardness of the discrete logarithm problem in the generic group model given by Definition~\ref{def-hard-DLog-ggm} follows immediately from
Theorem~\ref{eff-hard-DLog-ggm}.

\begin{corollary}\label{hard-DLog-ggm}
The discrete logarithm problem is hard in the generic group model.
\qed
\end{corollary}

We are interested in the instantiation of the generic group in the discrete logarithm problem.
Thus, it is convenient to define the hardness of the discrete logarithm problem relative to a specific family of encoding functions.

\begin{definition}\label{def-hard-DLog-rel}
Let $\{\sigma_n\}_{n\in\N^+}$ be a family of encoding functions. 
We say that \emph{the discrete logarithm problem is hard relative to $\{\sigma_n\}_{n\in\N^+}$}
if for all polynomial-time generic algorithms $\mathcal{A}$ and all $d\in\N^+$ there exists $N\in\N^+$ such that, for all $n\ge N$,
\begin{equation*}
  \Prob[\mathsf{DLog}_{\mathcal{A}}(n,\sigma)=1]\le\frac{1}{n^d}.
\end{equation*}
\qed
\end{definition}

The corresponding effective hardness notion is defined as follows.

\begin{definition}\label{def-eff-hard-DLog-rel}
Let $\{\sigma_n\}_{n\in\N^+}$ be a family of encoding functions. 
We say that \emph{the discrete logarithm problem is effectively hard relative to $\{\sigma_n\}_{n\in\N^+}$}
if there exists a computable function $f\colon\N^+\times\N^+\to\N^+$ such that, for all $i,d,n\in\N^+$, if $n\ge f(i,d)$ then
\begin{equation*}
  \Prob[\mathsf{DLog}_{\mathcal{A}_i}(n,\sigma)=1]\le\frac{1}{n^d}.
\end{equation*}
\qed
\end{definition}

In a similar manner to the proof of Theorem~\ref{FDH-computable}, we can prove the following theorem.

\begin{theorem}[Main result III]\label{DLog-computable}
There exists a computable family of encoding functions relative to which the discrete logarithm problem is effectively hard.
\end{theorem}

\begin{proof}%
First, by Theorem~\ref{eff-hard-DLog-ggm} there exists a computable function $f\colon\N^+\times\N^+\to\N^+$ such that, for all $i,d,n\in\N^+$, if $n\ge f(i,d)$ then
\begin{equation*}
  \frac{1}{\#\mathsf{Encf}_{n}}\sum_{\sigma\in\mathsf{Encf}_{n}}
  \Prob[\mathsf{DLog}_{\mathcal{A}_i}(n,\sigma)=1]\le\frac{1}{n^d}.
\end{equation*}
It follows from Lemma~\ref{lemma1} that, for all $i,d,n\in\N^+$, if $n\ge f(i,2d)$ then
\begin{equation}\label{keyineqnn_ggm}
\begin{split}
  \#\left\{\sigma\in\mathsf{Encf}_{n}\biggm|\Prob[\mathsf{DLog}_{\mathcal{A}_i}(n,\sigma)=1]>\frac{1}{n^d}\right\}
  <\frac{\#\mathsf{Encf}_{n}}{n^d}.
\end{split}
\end{equation}

In order to apply the method of algorithmic randomness, i.e., Theorem~\ref{exercise-modified},
for each $i,d,n\in\N^+$ we define a subset $\osg{C_{i,d,n}}$ of $\mathsf{Encf}^\infty$ as the set of all families $\{\sigma_n\}_{n\in\N^+}$ of encoding functions such that
\begin{equation}\label{eq-sn_ggm}
  \Prob[\mathsf{DLog}_{\mathcal{A}_i}(n,\sigma_n)=1]>\frac{1}{n^d}.
\end{equation}
Namely, we define a subset $C_{i,d,n}$ of $\mathsf{Encf}^*$ as the set of all finite families $(\sigma_1,\dots,\sigma_n)$ of encoding functions
where only $\sigma_n$ is required to satisfy the inequality~\eqref{eq-sn_ggm}.
Since $C_{i,d,n}$ is a prefix-free set for every $i,d,n\in\N^+$, it follows from (i) of Proposition~\ref{Lebesgue-outer_ggm} and \eqref{keyineqnn_ggm} that,
for each $i,d,n\in\N^+$, if $n\ge f(i,2d)$ then
\begin{equation}\label{keyineqnidn_ggm}
  \Lm{\osg{C_{i,d,n}}}=\sum_{s\in C_{i,d,n}}\abs{I(s)}<\frac{1}{n^d}.
\end{equation}

We choose a particular computable bijection
\begin{equation*}
  \varphi\colon\N^+\to\{\,(i,d)\mid i\in\N^+\;\&\;d\ge 2\,\},
\end{equation*}
and define $(\varphi_1(m),\varphi_2(m))=\varphi(m)$.
We then define a computable function $g\colon\N^+\to\N^+$ by
$g(m)=\{f(\varphi_1(m),2\varphi_2(m))+1\}^{m+1}$.
For each $m\in\N^+$, we define a subset $C_{m}$ of $\X$ by
\begin{equation}\label{C_m_ggm}
  C_m=\bigcup_{n=g(m)}^\infty C_{\varphi_1(m),\varphi_2(m),n}.
\end{equation}
It follows from (iii) of Proposition~\ref{Lebesgue-outer_ggm}, \eqref{keyineqnidn_ggm}, and Lemma~\ref{lemma2} that, for each $m\in\N^+$,
\begin{equation}\label{LC_mle2m_ggm}
\begin{split}
  \Lm{\osg{C_m}}\le\sum_{n=g(m)}^{\infty}\Lm{\osg{C_{\varphi_1(m),\varphi_2(m),n}}}
  <\sum_{n=g(m)}^{\infty}\frac{1}{n^{\varphi_2(m)}}\le\frac{2}{g(m)}\le\frac{1}{2^{m}}.
\end{split}
\end{equation}
We then define $C$
by
\begin{equation}\label{C_ggm}
  C=\bigcup_{m=1}^{\infty}C_m.
\end{equation}
Therefore, using (iii) of Proposition~\ref{Lebesgue-outer_ggm},
\begin{equation}
  \Lm{\osg{C}}\le\sum_{m=1}^{\infty}\Lm{\osg{C_m}}<\sum_{m=1}^{\infty}\frac{1}{2^{m}}=1.
\end{equation}

Next we show that $C$ is an r.e.~subset of $\mathsf{Encf}^*$.
It is easy to see that, given $i$, $d$, and $n$, one can decide the finite subset $C_{i,d,n}$ of $\mathsf{Encf}^*$,
since the dyadic rational $\Prob[\mathsf{DLog}_{\mathcal{A}_i}(n,\sigma)=1]$ is computable,
given $i$, $n$, and an encoding function $\sigma$ into $n$ bitstrings.
Thus, since $\varphi$ and $g$ are computable functions, it follows from \eqref{C_m_ggm} and \eqref{C_ggm} that $C$ is an r.e.~subset of $\mathsf{Encf}^*$.

We then show that $\Lm{\osg{C}}$ is a computable real.
For each $k\in\N$, we define a finite subset $D_k$ of
$C$
by
\begin{equation*}
  D_k=\bigcup_{m=1}^{k}\bigcup_{n=g(m)}^{g(m)2^{k}-1} C_{\varphi_1(m),\varphi_2(m),n}.
\end{equation*}
Given $k\in\N$, one can decides the finite set $D_k$,
since $\varphi$ and $g$ are computable functions and moreover one can decide the finite set $C_{i,d,n}$, given $i$, $d$, and $n$.
Therefore, given $k\in\N$, one can calculate the dyadic rational $\Lm{\osg{D_k}}$ based on (i) of Proposition~\ref{Lebesgue-outer_ggm}.
On the other hand, note that
\begin{equation*}
  C\setminus D_k\subset\left(\bigcup_{m=1}^{k}\bigcup_{n=g(m)2^{k}}^{\infty} C_{\varphi_1(m),\varphi_2(m),n}\right)\cup\bigcup_{m=k+1}^{\infty}C_m.
\end{equation*}
Thus, using (ii) and (iii) of Proposition~\ref{Lebesgue-outer_ggm}, \eqref{keyineqnidn_ggm},
Lemma~\ref{lemma2}, and \eqref{LC_mle2m_ggm} we see that, for each $k\in\N$,
\begin{align*}
  \Lm{\osg{C\setminus D_k}}
  &\le\sum_{m=1}^{k}\sum_{n=g(m)2^{k}}^{\infty}\Lm{\osg{C_{\varphi_1(m),\varphi_2(m),n}}}+\sum_{m=k+1}^{\infty}\Lm{\osg{C_m}} \\
  &<\sum_{m=1}^{k}\frac{2}{g(m)2^k}+\sum_{m=k+1}^{\infty}\frac{1}{2^m}
  \le\sum_{m=1}^{k}\frac{1}{2^{m+k}}+\frac{1}{2^{k}}
  <\frac{1}{2^{k-1}}.
\end{align*}
Therefore, since $\osg{C}=\osg{D_{k+1}}\cup\osg{C\setminus D_{k+1}}$, using (ii) and (iii) of Proposition~\ref{Lebesgue-outer_ggm} we have
$\abs{\Lm{\osg{C}}-\Lm{\osg{D_{k+1}}}}\le\Lm{\osg{C\setminus D_{k+1}}}\le 2^{-k}$
for each $k\in\N$.
Hence, $\Lm{\osg{C}}$ is a computable real.

Now, it follows from Theorem~\ref{exercise-modified} that there exists a computable family $\{\sigma_n\}_{n\in\N^+}$ of encoding functions which is not in $\osg{C}$.
Let $i,d,n\in\N^+$ with $n\ge g(\varphi^{-1}(i,d+1))$.
We then define $m=\varphi^{-1}(i,d+1)$, i.e., $\varphi(m)=(i,d+1)$.
Since $\{\sigma_n\}_{n\in\N^+}\notin\osg{C}$ and $n\ge g(m)$,
it follows from \eqref{C_ggm} and \eqref{C_m_ggm} that $\{\sigma_n\}_{n\in\N^+}\notin\osg{C_{\varphi_1(m),\varphi_2(m),n}}=\osg{C_{i,d+1,n}}$.
Therefore, we see that the family $\{\sigma_n\}_{n\in\N^+}$ of encoding functions satisfies that
$\Prob[\mathsf{DLog}_{\mathcal{A}_i}(n,\sigma_n)=1]\le 1/n^{d+1}<1/n^{d}$ for each $n\in\N^+$.
Thus, since the mapping $\N^+\times\N^+\ni (i,d)\mapsto g(\varphi^{-1}(i,d+1))$ is a computable function,
it follows from Definition~\ref{def-eff-hard-DLog-rel} that the discrete logarithm problem is effectively hard relative to $\{\sigma_n\}_{n\in\N^+}$.
\end{proof}

\begin{corollary}\label{DLog-computable-noneff}
There exists a computable family of encoding functions relative to which the discrete logarithm problem is hard.
\end{corollary}

\begin{proof}
The result follows immediately from Theorem~\ref{DLog-computable}.
\end{proof}

\section{The Diffie-Hellman problem}
\label{CDH}

In this section we consider the hardness of the computational Diffie-Hellman (CDH) problem in the generic group model.
For the CDH problem we can show the analogues of all the results about the discrete logarithm problem shown in the preceding sections. 
In this section, in particular we present the analogue of Theorem~\ref{DLog-computable}
for the CDH problem without proof.

We first recall the analogue of Theorem~\ref{Shoup} for the CDH problem.
We thus consider the following experiment for a polynomial-time generic algorithm $\mathcal{A}$,
a parameter $n$,
and a positive integer $N\le 2^n$:

\begin{quote}
\textbf{The computational Diffie-Hellman experiment $\mathsf{CDH}_{\mathcal{A}}(n,N)$:}
\vspace*{-1.5mm}
\renewcommand{\labelenumi}{\arabic{enumi}.}
\emph{
\begin{enumerate}
\item Generate an encoding function $\sigma$ into $n$ bitstrings uniformly.
\item Generate $x\in\Z_{N}$ uniformly.
\item Generate $y\in\Z_{N}$ uniformly.
\item The output of the experiment is defined to be $1$ if $\mathcal{A}(N;\sigma(1),\sigma(x),\sigma(y))=\sigma(xy)$ and $0$ otherwise.
\end{enumerate}
}
\end{quote}

Shoup~\cite{S97} showed the following lower bound for the complexity of the CDH problem in the generic group model,
which is the analog of Theorem~\ref{Shoup}.

\begin{theorem}[Shoup~\cite{S97}]\label{Shoup-CDH}
There exists $C\in\N^+$ such that, for every generic algorithm $\mathcal{A}$, $n\in\N^+$, and $N$ with $2\le N\le 2^n-1$,
\begin{equation*}
  \Prob[\mathsf{CDH}_{\mathcal{A}}(n,N)=1]\le\frac{C m^2}{p},
\end{equation*}
where $p$ is the largest prime divisor of $N$
and $m$ is the maximum number of the oracle queries among all the computation paths of $\mathcal{A}$.  
\qed
\end{theorem}

Now, consider the following experiment for a polynomial-time generic algorithm $\mathcal{A}$, a parameter $n$,
and an encoding function $\sigma$ into $n$ bitstrings:

\begin{quote}
\textbf{The computational Diffie-Hellman experiment $\mathsf{CDH}_{\mathcal{A}}(n,\sigma)$:}
\vspace*{-1.5mm}
\renewcommand{\labelenumi}{\arabic{enumi}.}
\emph{
\begin{enumerate}
\item Generate an $n$-bit prime $p$ uniformly.
\item Generate $x\in\Z_{p}$ uniformly.
\item Generate $y\in\Z_{p}$ uniformly.
\item The output of the experiment is defined to be $1$ if $\mathcal{A}(p;\sigma(1),\sigma(x),\sigma(y))=\sigma(xy)$ and $0$ otherwise.
\end{enumerate}
}
\end{quote}

Then the effective hardness of the CDH problem relative to a specific family of encoding functions is defined as follows.

\begin{definition}\label{def-eff-hard-CDH-rel}
Let $\{\sigma_n\}_{n\in\N^+}$ be a family of encoding functions. 
We say that \emph{the CDH problem is effectively hard relative to $\{\sigma_n\}_{n\in\N^+}$}
if there exists a computable function $f\colon\N^+\times\N^+\to\N^+$ such that, for all $i,d,n\in\N^+$, if $n\ge f(i,d)$ then
\begin{equation*}
  \Prob[\mathsf{CDH}_{\mathcal{A}_i}(n,\sigma)=1]\le\frac{1}{n^d}.
\end{equation*}
\qed
\end{definition}

Based on Theorem~\ref{Shoup-CDH}, we can show the following analogue of Theorem~\ref{DLog-computable} in the same manner as the proof of Theorem~\ref{DLog-computable}.

\begin{theorem}[Main result IV]\label{CDH-computable}
There exists a computable family of encoding functions relative to which the CDH problem is effectively hard.
\qed
\end{theorem}

\section{Polynomial-time effective hardness}
\label{conclusion-ggm}

In Section~\ref{discussion} we have demonstrated the importance of the effective security notions in modern cryptography.
The replacement of the conventional security notions of cryptographic schemes by the corresponding effective security notions results in
the replacement of the conventional hardness notions of computational problems, which are used as computational assumptions
to prove the security of the cryptographic schemes, by the corresponding effective hardness notions, as we have seen in Section~\ref{discussion}
where the RSA assumption in Theorem~\ref{RSA-FDH} is replaced by the effective RSA assumption in Theorem~\ref{effective_RSA-FDH}.
In addition, we have been able to prove the main results given in the previous two section, Theorems~\ref{DLog-computable} and \ref{CDH-computable},
by converting Shoup's original results about the lower bounds of the complexity, Theorems~\ref{Shoup} and \ref{Shoup-CDH},
into the form of effective hardness.
Thus, the effective hardness notions introduced in the previous two sections are useful and considered to be
a natural alternative to the conventional hardness notions of computational problems in modern cryptography.

Ultimately,
it would seem more natural to require that
the functions $f\colon\N^+\times\N^+\to\N^+$ in Definitions~\ref{def-eff-hard-DLog-ggm}, \ref{def-eff-hard-DLog-rel}, and \ref{def-eff-hard-CDH-rel} are
\emph{polynomial-time computable} rather than simply computable.
We call this type of effective hardness \emph{polynomial-time effective hardness}.
In Theorem~\ref{eff-hard-DLog-ggm} we have shown that the discrete logarithm problem is effectively hard in the generic group model.
In the proof of Theorem~\ref{eff-hard-DLog-ggm},
the
function $f$ has the form $f(i,d)=\max\{(2i+d+1)^2,2C\}$.
This is a polynomial-time computable function.
Thus, the proof of Theorem~\ref{eff-hard-DLog-ggm} actually shows that \emph{the discrete logarithm problem is polynomial-time effectively hard in the generic group model}.

Conjecture~\ref{conjecture-future} below is a polynomial-time effective version of Theorem~\ref{DLog-computable},
which states that the discrete logarithm problem is effectively hard in the standard model for \emph{some} finite cyclic group
such that the group operations are polynomial-time computable.
In the future, it would be challenging to determine whether Conjecture~\ref{conjecture-future} (or its appropriate modification) holds
for some computational assumption $\mathsf{COMP}$ which seems
weaker
than the hardness of the discrete logarithm problem itself.

\begin{conjecture}\label{conjecture-future_ggm}
Under the assumption $\mathsf{COMP}$,
there exists a \emph{polynomial-time computable} family of encoding functions (or a \emph{polynomial-time computable} family of families of encoding functions)
relative to which the discrete logarithm problem is \emph{polynomial-time effectively} hard.
\qed
\end{conjecture}

\section*{Acknowledgments}

This work was supported
by JSPS KAKENHI Grant Number 23340020
and by the Ministry of Economy, Trade and Industry of Japan.


\begin{thebibliography}{99}

\bibitem{BR93}
M. Bellare and P. Rogaway, Random oracles are practical: a paradigm for designing efficient protocols,
Proceedings of the 1st ACM Conference on Computer and Communications Security, ACM, pp.62--73, 1993.


\bibitem{BBP04}
M. Bellare, A. Boldyreva, and A. Palacio, An uninstantiable random-oracle-model scheme for a hybrid-encryption problem,
\emph{Proc}.~EUROCRYPT 2004, Lecture Notes in Computer Science, Springer-Verlag, Vol.3027, pp.171--188, 2004.

\bibitem{BMeN11}
L. Bienvenu, W. Merkle, and A. Nies, Solovay functions and $K$-triviality,
Proceedings of the 28th Symposium on Theoretical Aspects of Computer Science (STACS 2011), pp.452--463, 2011.

\bibitem{BMiN12}
V. Brattka, J. Miller, and A. Nies, ``Randomness and differentiability,'' preprint, 2012.

%
%

%
%

%
%
%

%
%
%

%
%

%
%

%
%

\bibitem{CGH04}
R. Canetti, O. Goldreich, and S. Halevi, ``The random oracle methodology, revisited,'' \emph{J. ACM}, vol.\ 51, pp.~557--594, 2004.

\bibitem{C66}
G. J. Chaitin, ``On the length of programs for computing finite binary sequences,'' \emph{J. Assoc. Comput. Mach.}, vol.\ 13, pp.~547--569, 1966.

\bibitem{C75}
G. J. Chaitin, ``A theory of program size formally identical to information theory,'' \emph{J. Assoc. Comput. Mach.}, vol.\ 22, pp.~329--340, 1975.

%
%

\bibitem{C87b}
G. J. Chaitin, \emph{Algorithmic Information Theory}. Cambridge University Press, Cambridge, 1987.

%
%
%

\bibitem{D02}
A. W. Dent, Adapting the weaknesses of the random oracle model to the generic group model,
\emph{Proc}.~ASIACRYPT 2002, Lecture Notes in Computer Science, Springer-Verlag, Vol.2501, pp.100--109, 2002.

%
%
%
%
%

\bibitem{DH10}
R. G. Downey and D. R. Hirschfeldt, \emph{Algorithmic Randomness and Complexity}. Springer-Verlag, New York, 2010.

%
%

\bibitem{FLRSST10}
M. Fischlin, A Lehmann, T. Ristenpart, T. Shrimpton, M. Stam, and S. Tessaro, Random oracles with(out) programmability,
\emph{Proc}.~ASIACRYPT 2010, Lecture Notes in Computer Science, Springer-Verlag, Vol.6477, pp.303--320, 2010.

%
%

\bibitem{G01}
O. Goldreich, \emph{Foundations of Cryptography: Volume 1 -- Basic Tools}. Cambridge University Press, New York, 2001.

\bibitem{G04}
O. Goldreich, \emph{Foundations of Cryptography: Volume 2 -- Basic Applications}. Cambridge University Press, New York, 2004.

\bibitem{IR88}
R. Impagliazzo and S. Rudich, Limits on the provable consequences of one-way permutations,
\emph{Proc}.~CRYPTO'88, Lecture Notes in Computer Science, Springer-Verlag, Vol.403, pp.8--26, 1990.
%

\bibitem{KL07}
J. Katz and Y. Lindell, \emph{Introduction to Modern Cryptography}. Chapman \& Hall/CRC Press, 2007.

\bibitem{Kol65}
A. N. Kolmogorov, ``Three approaches to the quantitative definition of information,'' \emph{Problems Inform. Transmission}, vol.\ 1, no.\ 1, pp.~1--7, 1965.

%
%

\bibitem{LN09}
G. Leurent and P. Q. Nguyen, How risky is the random-oracle model?
\emph{Proc}.~CRYPTO 2009, Lecture Notes in Computer Science, Springer-Verlag, Vol.5677, pp.445--464, 2009.

%
%
%

\bibitem{M66}
P. Martin-L\"{o}f, ``The definition of random sequences,'' \emph{Information and Control}, vol.\ 9, pp.~602--619, 1966.

\bibitem{MW98}
U. Maurer and S. Wolf, Lower bounds on generic algorithms in groups,
\emph{Proc.}~EUROCRYPT'98, Lecture Notes in Computer Science, Springer-Verlag, Vol.1403, pp.72--84, 1998.

\bibitem{M05}
U. Maurer, Abstract models of computation in cryptography,
%
\emph{Proc.}~Cryptography and Coding 2005, Lecture Notes in Computer Science, Springer-Verlag, Vol.3796, pp.1--12, 2005.

\bibitem{MY08}
J. Miller and L. Yu, ``On initial segment complexity and degrees of randomness,'' \emph{Trans. Amer. Math. Soc.}, vol.\ 360, pp.~3193--3210, 2008.

\bibitem{MNO11}
D. Moriyama, R. Nishimaki and T. Okamoto, \emph{Theory of Public-Key Cryptography}.
Industrial and Applied Mathematics Series Vol.2. JSIAM, Kyoritsu Shuppan Co., Ltd., Tokyo, 2011. In Japanese.

\bibitem{N09}
A. Nies, \emph{Computability and Randomness}. Oxford University Press, Inc., New York, 2009.

%
%
%
%

\bibitem{PR89}
M. B. Pour-El and J. I. Richards, \emph{Computability in Analysis and Physics}. Perspectives in Mathematical Logic, Springer-Verlag, Berlin, 1989.

%
%
%
%
%

%
%

\bibitem{Sch71}
C.-P. Schnorr, ``A unified approach to the definition of a random sequence,'' \emph{Mathematical Systems Theory}, vol.\ 5, pp.~246--258, 1971.

\bibitem{Sch73}
C.-P. Schnorr, ``Process complexity and effective random tests,'' \emph{J.\ Comput.\ System Sci.}, vol.\ 7, pp.~376--388, 1973.

\bibitem{S97}
V. Shoup, Lower bounds for discrete logarithms and related problems,
\emph{Proc}.~EUROCRYPT'97, Lecture Notes in Computer Science, Springer-Verlag, Vol.1233, pp.256--266, 1997.

\bibitem{Solom64}
R. J. Solomonoff, ``A formal theory of inductive inference. Part I and Part II,'' \emph{Inform. and Control}, vol.\ 7, pp.~1--22, 1964; vol.\ 7, pp.~224--254, 1964.
  
\bibitem{Sol75}
R. M. Solovay, ``Draft of a paper (or series of papers) on Chaitin's work ... done for the most part during the period of Sept.--Dec.\ 1974,''
unpublished manuscript, IBM Thomas J.\ Watson Research Center, Yorktown Heights, New York, May 1975, 215 pp.

%
%
%
%
%
%
%
%
%

%
%
%
%
%
%
%
%
%
%
%
%
%

%
%
%
%
%
%
%
%
%
%
%
%
%
%
%

%
%
%
%

\bibitem{W00}
K. Weihrauch, \emph{Computable Analysis}. Springer-Verlag, Berlin, 2000.

%
%
%
%

\end{thebibliography}
\end{document}